\documentclass[journal,twoside,web]{ieeecolor}

\usepackage{tmi,cite,graphicx,mathtools,amsmath,amssymb,amsfonts,amsthm,bm,hyperref,algorithm,algorithmic,soul,color}

\makeatletter
\let\@@pmod\pmod
\DeclareRobustCommand{\pmod}{\@ifstar\@pmods\@@pmod}
\def\@pmods#1{\mkern4mu({\operator@font mod}\mkern 6mu#1)}
\makeatother

\hypersetup{
    colorlinks=true,
    linkcolor=black,
    filecolor=black,      
    urlcolor=black
    }

\newtheoremstyle{mystyle}
  {}
  {}
  {\itshape}
  {}
  {\bfseries}
  {.}
  { }
  {}

\theoremstyle{mystyle}
\newtheorem{definition}{Definition}
\newtheorem{theorem}{Theorem}
\newtheorem{lemma}{Lemma}

\newcommand{\appropto}{\mathrel{\vcenter{
  \offinterlineskip\halign{\hfil$##$\cr
    \propto\cr\noalign{\kern2pt}\sim\cr\noalign{\kern-2pt}}}}}

\def\BibTeX{{\rm B\kern-.05em{\sc i\kern-.025em b}\kern-.08em T\kern-.1667em\lower.7ex\hbox{E}\kern-.125emX}}
\begin{document}
\title{Venc Design and Velocity Estimation for Phase Contrast MRI}
\author{Shen Zhao,  Rizwan Ahmad, and Lee  C. Potter, \IEEEmembership{Senior Member, IEEE}
\thanks{This manuscript was submitted on September 10, 2021 and revised on February 17, 2022, May 6, 2022, and July 10, 2022. This work was supported in part by NIH grants R01HL135489 and R01HL151697. }
\thanks{S.\ Zhao and L.\ Potter are with the Department of Electrical \& Computer Engineering, Ohio State University, Columbus, OH 43210 USA (email: zhao.1758@osu.edu, potter.36@osu.edu).}
\thanks{R.\ Ahmad is with the Department of Biomedical Engineering, Ohio State University, Columbus, OH 43210 USA (email: ahmad.46@osu.edu).}}
\maketitle
\begin{abstract}
In phase-contrast magnetic resonance imaging (PC-MRI), spin velocity contributes to the phase measured at each voxel. Therefore, estimating velocity from potentially wrapped phase measurements is the task of solving a system of noisy congruence equations. We propose \emph{Phase Recovery from Multiple Wrapped Measurements} (PRoM) as a fast, approximate maximum likelihood estimator of velocity from multi-coil data with possible amplitude attenuation due to dephasing. The estimator can recover the fullest possible extent of unambiguous velocities, which can greatly exceed 
twice the highest venc. The estimator uses all pairwise phase differences and the inherent correlations among them to minimize the estimation error. Correlations are directly estimated from multi-coil data without requiring knowledge of coil sensitivity maps, dephasing factors, or the actual per-voxel signal-to-noise ratio. Derivation of the estimator yields explicit probabilities of unwrapping errors and the probability distribution for the velocity estimate; this, in turn, allows for optimized design of the phase-encoded acquisition. These probabilities are also incorporated into spatial post-processing to further mitigate wrapping errors. Simulation, phantom, and in vivo results for three-point PC-MRI acquisitions validate the benefits of reduced estimation error, increased recovered velocity range, optimized acquisition, and fast computation. A phantom study at 1.5\,T demonstrates 48.5\% decrease in root mean squared error using PRoM with post-processing versus a conventional ``dual-venc” technique. Simulation and 3\,T in vivo results likewise demonstrate the proposed benefits.
\end{abstract}


\section{Introduction}
\label{sec:introduction}
\IEEEPARstart{P}{hase}-contrast magnetic resonance imaging (PC-MRI) is a quantitative, non-invasive technique to measure hemodynamics in vivo \cite{pelc1991phase}. PC-MRI also enables higher dimensional velocimetry such as 4D flow imaging \cite{markl20124d}. Spin velocity is encoded into the voxel phase via a time-varying gradient field.
Lowering the first moment of the encoding gradient extends the unaliased range but degrades the velocity-to-noise ratio (VNR). To address this issue, multi-point acquisitions such as ``dual-venc'' have been proposed, whereby a high-venc measurement is used to unwrap a potentially wrapped, but less noisy, low-venc velocity measurement \cite{Lee1995, Schnell2017}. 
In contrast, multiple (possibly wrapped) pairwise phase differences can be jointly processed \cite{Zhao2018, loecher2018velocity, Carrillo2019}, leading to a potentially larger unambiguous range of velocities and a reduced estimation error. Existing estimators do not guide the optimized design of the acquisition and employ a computationally expensive grid search or gradient descent iteration.

Building on our preliminary work \cite{Zhao2018}, we propose \emph{Phase Recovery from Multiple Wrapped Measurements} (PRoM), an approximate maximum likelihood estimator (MLE) for a set of linear congruence equations with additive correlated noise. The estimator is used in multi-point PC-MRI to process all pairwise phase differences jointly. The PRoM estimator first constructs a set of candidate tuples of wrapping integers. The probability that the true tuple of wrapping integers is in this set is arbitrarily close to $1$. For each candidate tuple, the corresponding candidate velocity is found without grid searching as a simple weighted combination of the noisy measurements plus the candidate wrapping. The final velocity estimate is chosen among the small set of candidate velocities to maximize the likelihood function. The approximate MLE explicitly accommodates both intra-voxel dephasing and the inherent noise correlation present among phase differences. The proposed estimator does not need the complex-valued coil sensitivities, the actual dephasing factors, or the actual per-voxel SNR. Instead, the computation of the phase difference error covariance automatically accounts for these effects directly from the complex-valued data from the set of phase-encoded images across all coils and all encodings.

The probability distribution of the velocity estimate from noisy data is derived, allowing for the optimized phase encoding design. The likelihoods of wrapping integers provided by the proposed estimator are leveraged in spatial post-processing to mitigate unwrapping errors. Additionally, the same estimation problem appears in other array processing applications, where PRoM extends current art \cite{wang2015maximum, WangNehorai2017} to provide: a fast, grid-free estimator; accommodation of correlated noise; and principled design of sparse array geometry. 

For validation, PRoM is applied to data sets acquired from simulation, 1.5\,T scan of a spinning phantom, and 3\,T in vivo scan of a healthy volunteer.

\section{Theory}
\label{sec:theory}

\subsection{Phase Encoding}
A time-varying magnetic gradient field may be used to encode spin velocity into image phase. Here, we consider encoding the velocity component in one direction. Consider a spin moving through a magnetic field, the Taylor series expansion of spin position $p(t) \in \mathbb{R}$ at $t=0$ yields
\begin{equation}
    p(t) = p(0) + \frac{v}{1!} t + \frac{a}{2!} t^2 + \cdots,
\end{equation}
where $v$ is instantaneous velocity and $a$ is acceleration. Let $\gamma$ be the gyromagnetic ratio, $B_0 $ be the main static magnetic field strength, and $g(t)\in \mathbb{R}$ be the time-varying magnetic field gradient. Then to first order approximation of $p(t)$ \cite{pelc1991phase}, the phase accumulated from $t=0$ to echo time TE is 
\begin{align}
    \phi    &= \int_0^\text{TE}  \gamma \left [ B_0 + g(t)p(t) \right] d t             \nonumber\\
            &\approx \overbrace{\gamma B_0  \text{TE} + \gamma p(0) \underbrace{\int_{0}^{\text{TE}} g(t) d t}_{m_0}}^{\phi_0} + \gamma v \underbrace{ \int_{0}^{\text{TE}} t g(t)  d t}_{m_1} \nonumber\\
            & = \phi_0 + \gamma m_1 v,
\end{align}
where $m_0$ and $m_1$ denote the zeroth and first moments of $g(t)$. The zeroth moment, $m_0$, encodes the spin position into phase and combines with $\gamma B_0  \text{TE}$ to make the background phase, $\phi_0$. The first moment, $m_1$, encodes spin velocity into phase.
So for $N_e$-point encoding and $N_c$ coils, the integral of all spins in a voxel yields the measurement
\begin{align}
    \label{eq: datamodel}
    \widetilde{y}_{\alpha \beta} = A_\alpha S_\beta e^{i (\phi_0 + \gamma m_{1\alpha} v )} + \delta_{\alpha \beta},
\end{align}
where $\alpha \in \{1, ..., N_e\}$ indexes over all encodings, and $\beta\in \{1, ..., N_c\}$ indexes over all coils.  The resulting signal amplitude is $A_\alpha \in \mathbb{R}$; $S_\beta \in \mathbb{C}$ is the coil sensitivity; $v$ is the resultant instantaneous velocity; $m_{1\alpha} \in \mathbb{R}$ is the first moment; $\delta_{\alpha \beta}\in \mathbb{C}$ is independent and identically distributed (i.i.d.) complex circularly symmetric Gaussian noise. This i.i.d. assumption can be aided by pre-whitening along the coil dimension. The existence of heterogeneous spin velocities and proton density can make $v$ in \eqref{eq: datamodel} differ from the mean moving spin velocity (e.g., partial volume effect) and can reduce the amplitude (e.g., intra-voxel dephasing effect). Generally, $A_\alpha$ decreases as $|m_{1\alpha}|$ increases \cite{o2008mri}.

Let $\widetilde{\bm{Y}}$ denote an ${N_e\times N_c}$
complex-valued measurement matrix with $(\alpha,\beta)^\text{th}$ entry  $\widetilde{y}_{\alpha \beta}$ for encoding $\alpha$ and coil $\beta$.  Observe that the unambiguous range of velocities for the $N_e$ encodings is
\begin{equation}
    \label{eq: Range Omega Prime}
    \Omega' = \text{LCM}\left(\frac{2\pi}{\gamma m_{11}}, \cdots, \frac{2\pi}{\gamma m_{1N_e}} \right),
\end{equation}
where $\text{LCM}(\cdot)$ denotes least common multiple, which is the smallest positive real number that is an integer multiple of all input numbers. 
It follows that $v$ and $v+ \Omega'$ are indistinguishable given data, $\widetilde{\bm{Y}}$. Considering noise, the MLE of $v$ involves $N_e+2N_c+2$ real-valued unknowns, $\left\{ v, \phi_0, A_1, \cdots, A_{N_e},|S_1|,\angle S_1, \cdots,|S_{N_c}|, \angle S_{N_c} \right\}$, and is a nonlinear least-squares fit to the data. Here, $\angle (\cdot)$ denotes angle of a complex number. The optimization task can be reduced to 
\begin{align}
    \label{eq:  MLE}
    \underset{\bm{s}}{\operatorname{argmax}}~  \frac{\bm{s}^\mathsf{H} \widetilde{\bm{R}} \bm{s}}{\bm{s}^\textsf{H}\bm{s}}~ \operatorname{s.t.} \|\bm{s}\| = 1,
\end{align}
where $\widetilde{\bm{R}}\in \mathbb{C}^{N_e \times N_e}$ and the ``steering vector'' $\bm{s}\in\mathbb{C}^{N_e\times 1}$ are
\begin{align}
    \widetilde{\bm{R}} &= \widetilde{\bm{Y}} \widetilde{\bm{Y}}^\mathsf{H} \nonumber\\
    \bm{s} &= \left[ A_1 e^{i  \gamma m_{11} v}, ... , A_{N_e} e^{i  \gamma m_{1N_e} v} \right]^{\intercal}.
\end{align}
Derivation of \eqref{eq:  MLE} is a straightforward extension of \cite[p.~288]{StoicaMoses} to accommodate
unequal amplitudes, $A_\alpha$.
Here, $(\cdot)^{\intercal}$ and $(\cdot)^\mathsf{H}$ denote transpose and conjugate transpose. 
For a given $v$, the amplitudes $ A_\alpha$ may be found by solving \eqref{eq:  MLE} as a generalized eigenvalue problem \cite[p.\,461]{van1996matrix}; yet, the optimization over $v$ nonetheless encounters a difficult cost surface with many local minima. A simple illustration of the sum of squared error versus velocity using a single-coil, symmetric three-point acquisition is shown in Fig.~\ref{fig: sum of squared error}. Data are simulated using $\left[\frac{|A_1S_1|}{\sigma(\delta_{11})}, \frac{|A_2S_1|}{\sigma(\delta_{21})},\frac{|A_3S_1|}{\sigma(\delta_{31})}\right] = [2.5,5,2.5]$ to represent a moderate 50\% loss of signal amplitude due to intra-voxel dephasing. Throughout the paper, we use same 50\% in simulation. The fit error at the true velocity, $v= 0$\,cm/s, is shown by the red diamond, and the global minimum at $v \approx -1.01$\,cm/s for the given noise realization is shown by the green marker. Competing local minima, shown as blue markers, may become the global minimum due to noise, thereby causing unwrapping errors. 
\begin{figure}[h!]
    \centering
    \includegraphics[width = \columnwidth]{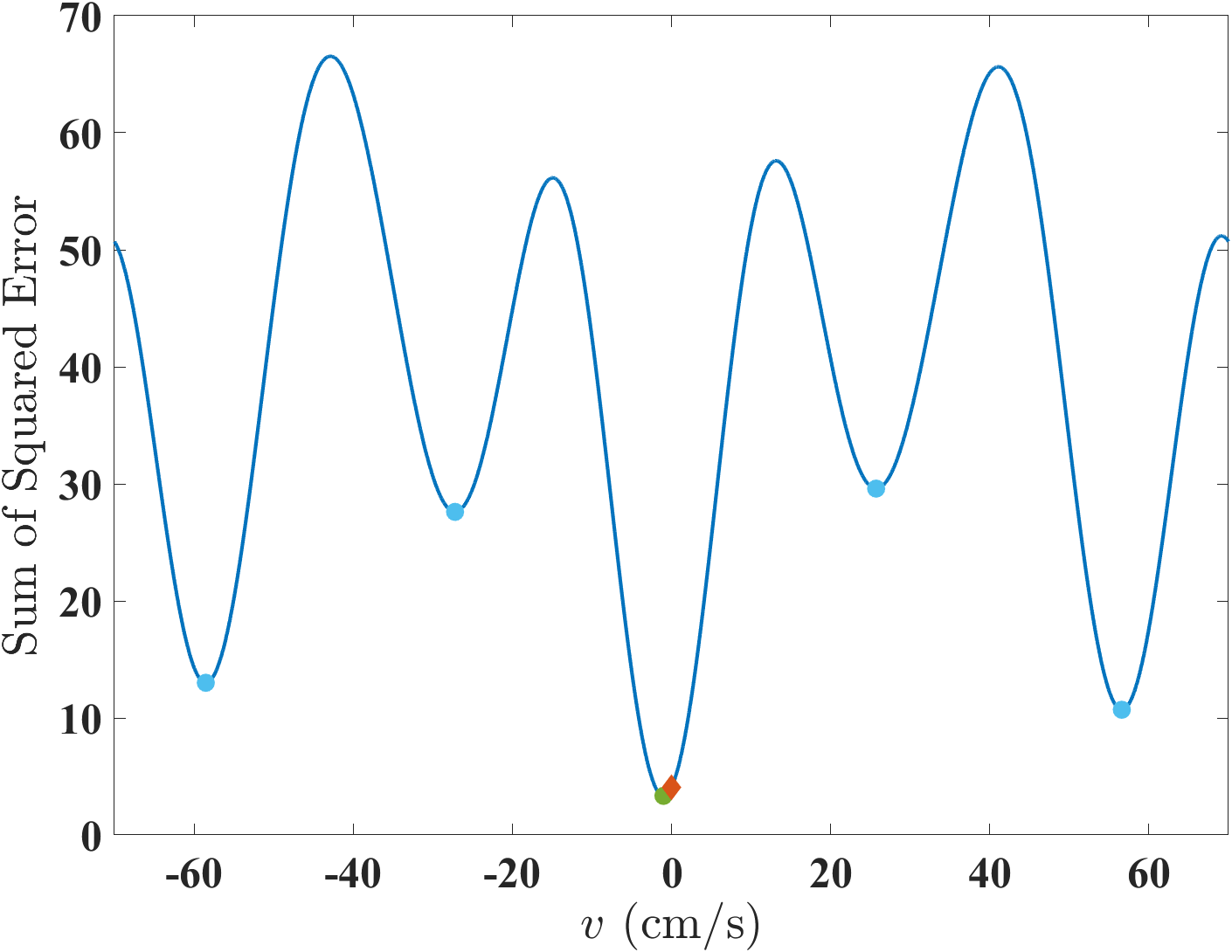}
    \caption{Sum of squared error versus velocity for a single coil, symmetric three-point acquisition. $\gamma[m_{11}, m_{12}, m_{13}] = \left[-\frac{\pi}{20},\frac{\pi}{70},\frac{\pi}{20}\right]$\,s/cm,  $\left[\frac{|A_1S_1|}{\sigma(\delta_{11})}, \frac{|A_2S_1|}{\sigma(\delta_{21})},\frac{|A_3S_1|}{\sigma(\delta_{31})}\right] = [2.5,5,2.5]$. The fit error to the noisy data at the true velocity $v=0$\,cm/s is shown by the red diamond. The global minimum at $v \approx -1.01$\,cm/s for the given noise realization is shown by the green marker. Competing local minima, shown as light blue markers, may become the global minimum due to noise, thereby causing unwrapping errors.}
    \label{fig: sum of squared error}
\end{figure}

\subsection{Estimation of Phase Noise Covariance}
\label{subsection: Estimation of Phase Noise Covariance}
From \eqref{eq:  MLE} we see that $\widetilde{\bm{R}}$ is a sufficient statistic for $v$ in \eqref{eq: datamodel}.
We depart from the multi-variate optimization in \eqref{eq:  MLE} to instead work with the phases of the off-diagonal entries of $\widetilde{\bm{R}}$, and we use amplitudes of $\widetilde{\bm{Y}}$ to construct the phase difference noise covariance matrix.
In so doing, we gain four advantages:
(1) characterization of the unambiguous set of estimated velocities;
(2) characterization of the probability of unwrapping errors;
(3) ability to design the encodings $\left[ m_{11}, ..., m_{1N_e} \right]$
to minimize mean squared estimation error subject to guarantees on the probability of unwrapping error and required unambiguous range; 
(4) fast, grid-free parameter estimator for velocity, $v$.
Further, the proposed surrogate estimator is asymptotically efficient, providing a good MLE approximation. In this section, we derive an approximate covariance matrix for the phase measurements in $\widetilde{\bm{R}}$ to lay the groundwork for building the proposed estimator of $v$.

Observe that $\widetilde{\bm{R}} = \widetilde{\bm{Y}} \widetilde{\bm{Y}}^\mathsf{H}$ performs coil combining and phase differencing. Denote the $(a,b)^\text{th}$ entry of $\widetilde{\bm{R}}$ as $\widetilde{r}_{ab}$, so
\begin{equation}
    \widetilde{r}_{ab} = \sum_\beta\widetilde{y}_{a\beta} \widetilde{y}_{b\beta}^*,   
\end{equation}
where $(\cdot)^*$ denotes complex conjugation. The noisy phase difference $\widetilde{\theta}_{ab}$ for encodings $a > b \in \{ 1,\cdots, N_e\}$ is 
\begin{equation}
    \label{eq: phasediff}
    \widetilde{\theta}_{ab} = \angle \widetilde{r}_{ab}.
\end{equation}
This phase differencing results in venc value
\begin{align}
    \label{eq: vencs}
    \text{venc}_{ab} = \frac{\pi}{\gamma ( m_{1a} - m_{1b} )}.
\end{align}
We have $\frac{N_e(N_e-1)}{2}$ such combinations. Throughout the paper, 
we let $\text{venc}_{ab}$ have units cm/s. Multiplying the vencs with phases, we obtain a (possibly wrapped) noisy velocity $\widetilde{v}_{ab}$
\begin{equation}
    \label{eq: remainder notation}
    \widetilde{v}_{ab} = \frac{\widetilde{\theta}_{ab}}{\pi}\text{venc}_{ab}.
\end{equation}
Phases $\widetilde{\theta}_{ab}$ are unambiguous on any interval of length $2\pi$, and for convenience we use $[0,2\pi)$; so, $\widetilde{v}_{ab} \in [0, 2\text{venc}_{ab})$. Thus, we have a set of noisy congruence equations for $v$:
\begin{equation}
    \label{eq: congruence}
    \widetilde{v}_{ab} \equiv v + n_{ab} \mod 2\text{venc}_{ab},
\end{equation}
where $n_{ab}$ is additive zero mean noise. Define $k_{ab}$ as the wrapping integer for $\widetilde{v}_{ab}$. Momentarily assume that the true wrapping integer, $k_{ab}$, is known; then, we can rewrite \eqref{eq: congruence} as a set of linear equations:
\begin{equation}
    \label{eq: equality formulation}
    \widetilde{\bm{v}} + \bm{k}\circ2\textbf{venc} = v +\bm{n},
\end{equation}
where $\circ$ is the Hadamard (element-wise) product, and 
\begin{align}
    \widetilde{\bm{v}}  &= \left[ \widetilde{v}_{21},\cdots, \widetilde{v}_{N_e(N_e-1)}  \right]^\intercal \nonumber\\
    \bm{k}              &= \left[k_{21},             \cdots, k_{N_e(N_e-1)}              \right]^\intercal \nonumber\\
    \textbf{venc}       &= \left[ \text{venc}_{21},  \cdots, \text{venc}_{N_e(N_e-1)}    \right]^\intercal \nonumber\\
    \bm{n}              &= \left[ n_{21},            \cdots, n_{N_e(N_e-1)}              \right]^\intercal. 
    \label{eq: definition}
\end{align}
From the Chinese remainder theorem, the smallest $\Omega > 0$ such that $v+\Omega$ satisfies \eqref{eq: congruence} is
\begin{equation}
    \label{eq: Range Omega}
    \Omega = \text{LCM}(2\textbf{venc}).
\end{equation}
This also implies that $\Omega$ is the smallest repetition period of $v$ to construct the same $\widetilde{\bm{R}}$. Recall from \eqref{eq: Range Omega Prime} that $\Omega'$ is a repetition period of $v$ to construct the same $\widetilde{\bm{Y}}$, as well as the same $\widetilde{\bm{R}}$. So, $\Omega'$ is an integer multiple of $\Omega$.

Similar to the assumption $\widetilde{\theta}_{ab} \in [0, 2\pi)$, we assume $v \in [0,\Omega)$ for convenience of derivation. This interval could also be shifted to $\left[-\frac{\Omega}{2},\frac{\Omega}{2}\right)$, for example, for bidirectional flow in MRI.

Let $\widehat{v} \in [0,\Omega)$ be the linear unbiased estimator of $v$ with smallest root mean squared error (RMSE). To compute $\widehat{v}$ from the noisy $\widetilde{\bm{v}}$ in \eqref{eq: equality formulation}, we only need the covariance matrix of the noise in the remainders,  $\bm{\Sigma}\left(\bm{n}\right)$. Define $\sigma$ to be the standard deviation of the i.i.d. noise $\delta_{\alpha \beta}$ in \eqref{eq: datamodel}. The mean and variance of $\widetilde{r}_{ab}$ are \cite{ODonoughue2012}:
\begin{align}
	\label{eq: prodmean}
	\mathbb{E}\left(\widetilde{r}_{ab}\right) &=  \sum_\beta A_a A_b  e^{i\gamma (m_{1a}-m_{1b}) v} |S_\beta|^2\\
	\label{eq: prodvar}
	\sigma^2 \left(\widetilde{r}_{ab}\right) &= \sum_\beta\left (\sigma^2 (A_a^2 + A_b^2)|S_\beta|^2 + \sigma^4\right).
\end{align}
Then from \eqref{eq: prodmean}-\eqref{eq: prodvar}, the squared signal-to-noise ratio (SNR), or Rician factor, for $\widetilde{r}_{ab}$ is
\begin{align}
    \operatorname{SNR}^2 \left(\widetilde{r}_{ab}\right) \nonumber =& \frac{\left|\mathbb{E}\left(\widetilde{r}_{ab}\right)\right|^2}{\sigma^2\left(\widetilde{r}_{ab}\right)} \nonumber\\
    =& \frac{\left(\sum_\beta A_a A_b|S_\beta|^2\right)^2}{\sum_\beta\left(\sigma^2 (A_a^2 + A_b^2 )|S_\beta|^2 + \sigma^4\right) } \nonumber\\
    \label{eq:  squared SNR for sum of conjugate product}
    =& \frac{\left(\sum_\beta s_{a\beta} s_{b\beta}\right)^2}{\sum_\beta(s_{a\beta}^2 + s_{b\beta}^2 + 1)},
\end{align}
where $s_{\alpha \beta} = \frac{|A_\alpha S_\beta|}{\sigma} = \operatorname{SNR}\left(\widetilde{y}_{\alpha \beta}\right)$. As $\operatorname{SNR} \left(\widetilde{r}_{ab}\right) \rightarrow \infty$,  $\widetilde{\theta}_{ab}$ converges in distribution to a Gaussian random variable. On the contrary, as $\operatorname{SNR} \left(\widetilde{r}_{ab}\right) \rightarrow 0$,  $\widetilde{\theta}_{ab}$ converges in distribution to a uniformly distributed random variable on $[0, 2\pi)$. Because $\lim_{x\rightarrow 0} \frac{x}{\sin(x)} = 1$, and $\sigma \left(\widetilde{\theta}_{ab}\right)$ conditioned on no wrapping is inversely proportional to $\operatorname{SNR} \left(\widetilde{r}_{ab}\right)$, 
\begin{equation}
    \label{eq: varapprox}
    \sigma^2\left( \widetilde{\theta}_{ab} \right) 
    \approx \frac{1}{2\operatorname{SNR}^2 \left(\widetilde{r}_{ab}\right)} 
    = \frac{\sum_\beta(s_{a\beta}^2 + s_{b\beta}^2 + 1)}{2\left(\sum_\beta s_{a\beta} s_{b\beta}\right)^2}.
\end{equation}
Similarly, we can obtain covariance given no wrapping
\begin{align}
    \operatorname{cov} \left( \widetilde{\theta}_{ab}, \widetilde{\theta}_{cb}  \right) 
    &\approx \frac{\left|\operatorname{cov}\left( \widetilde{r}_{ab},\widetilde{r}_{cb} \right)\right|}{\sum_\beta 2A_aA_b^2A_c |S_\beta|^4} 
    = \frac{N_c}{\sum_\beta 2s_{b\beta}^2} \nonumber\\\label{eq: noshared}
    \operatorname{cov} \left( \widetilde{\theta}_{ab}, \widetilde{\theta}_{bd}  \right)
    & \approx \frac{-N_c}{\sum_\beta 2s_{b\beta}^2}  
\end{align}
and $\operatorname{cov} \left( \widetilde{\theta}_{ab}, \widetilde{\theta}_{cd}  \right)=0$
for two phase differences that do not share a common encoding.

In practice, it may be difficult to accurately estimate the noise power for each voxel, and thus hard to estimate $s_{\alpha \beta}$. To ameliorate estimation difficulty and use complex measurements $\widetilde{y}_{\alpha \beta}$ only, we further approximate and scale the variance and covariance:
\begin{align}
    \label{eq: approximated scaled sigma}
    \sigma^2\left( \widetilde{\theta}_{ab} \right) 
    &\approx \frac{\sum_\beta (s_{a\beta}^2+s_{b\beta}^2)}{2\left(\sum_\beta s_{a\beta}s_{b\beta}\right)^2} 
    \appropto  \frac{\sum_\beta \left(|\widetilde{y}_{a\beta}|^2+|\widetilde{y}_{b\beta}|^2\right)}{2\left(\sum_\beta |\widetilde{y}_{a\beta}||\widetilde{y}_{b\beta}|\right)^2}\\
    \label{eq: approximated scaled cov}
    \operatorname{cov} \left( \widetilde{\theta}_{ab}, \widetilde{\theta}_{cb}  \right) & \appropto \frac{N_c}{\sum_\beta 2|\widetilde{y}_{b\beta}|^2},
\end{align}
where $\appropto$ denotes ``approximately proportional to.''  Here, the scaling factors are the same for \eqref{eq: approximated scaled sigma} and \eqref{eq: approximated scaled cov}. Let 
\begin{equation}
    \widetilde{\bm{\theta}}  = \left[ \widetilde{\theta}_{21},\cdots, \widetilde{\theta}_{N_e(N_e-1)}  \right]^\intercal \nonumber.
\end{equation}
Then, with the elementwise approximations above, we can formulate the scaled approximated $\bm{\Sigma}\left(\widetilde{\bm{\theta}}\right)$ directly from observed voxel magnitudes. This scaling in \eqref{eq: approximated scaled sigma} and \eqref{eq: approximated scaled cov} does not affect the estimator $\widehat{v}$, as seen from \eqref{eq: cost} below. Finally, because the true covariance matrix is close to rank deficient, the element-wise approximations in \eqref{eq: approximated scaled sigma} and \eqref{eq: approximated scaled cov} can potentially violate the positive semi-definite property of a covariance matrix. Accordingly, we follow the approximation step by projection to the closest positive semi-definite matrix \cite{halmos1972positive}. This projection operator, $\Pi(\cdot)$, for a symmetric matrix $\mathbf{M}$ with eigen-decomposition $\mathbf{M} = \mathbf{V}\mathbf{S} \mathbf{V}^{-1}$ is given by
\begin{equation}
    \Pi(\mathbf{M}) = \mathbf{V} \max(\mathbf{S},0)\mathbf{V}^{-1},
\end{equation}
where $\max(\mathbf{S},0)$ is applied element-wise to the eigenvalues. 

To illustrate accuracy of covariance modeling, we adopt the cosine similarity metric, which is scale invariant. For modeled covariance $\Pi\left(\bm{\Sigma}\left(\widetilde{\bm{\theta}}\right)\right)$ and sample covariance $\widetilde{\bm{\Sigma}}\left(\bm{\theta}\right)$, the cosine similarity is 
\begin{equation}
    \frac{\text{Trace} \left( \Pi\left(\bm{\Sigma}\left(\widetilde{\bm{\theta}}\right)\right)^\mathsf{H}\widetilde{\bm{\Sigma}}\left(\bm{\theta}\right)\right)}
    {\left\|\Pi\left(\bm{\Sigma}\left(\widetilde{\bm{\theta}}\right)\right)\right\|_\mathsf{F} \left\|\widetilde{\bm{\Sigma}}\left(\bm{\theta}\right)\right\|_\mathsf{F}}.
\end{equation}
Results are computed for the case of a single-coil, symmetric encoding, and intra-voxel dephasing $2s_{11} = s_{21} = 2s_{31}$; $10^6$ random draws are used at each $s_{21}$ to provide a sample covariance matrix
and to compute the mean cosine similarity.
Fig.~\ref{fig: verify approximation} displays the similarity metric results,
from which we see excellent agreement for $s_{21} > 2$.  Thus, the covariance model is accurate at modest SNR.  Also shown in Fig.~\ref{fig: verify approximation} is the similarity metric for the (scaled) identity covariance model, which is implicitly employed when using least-squares estimation with phase differences \cite{Carrillo2019}.

\begin{figure}[h!]
    \centering
    \includegraphics[width = \columnwidth]{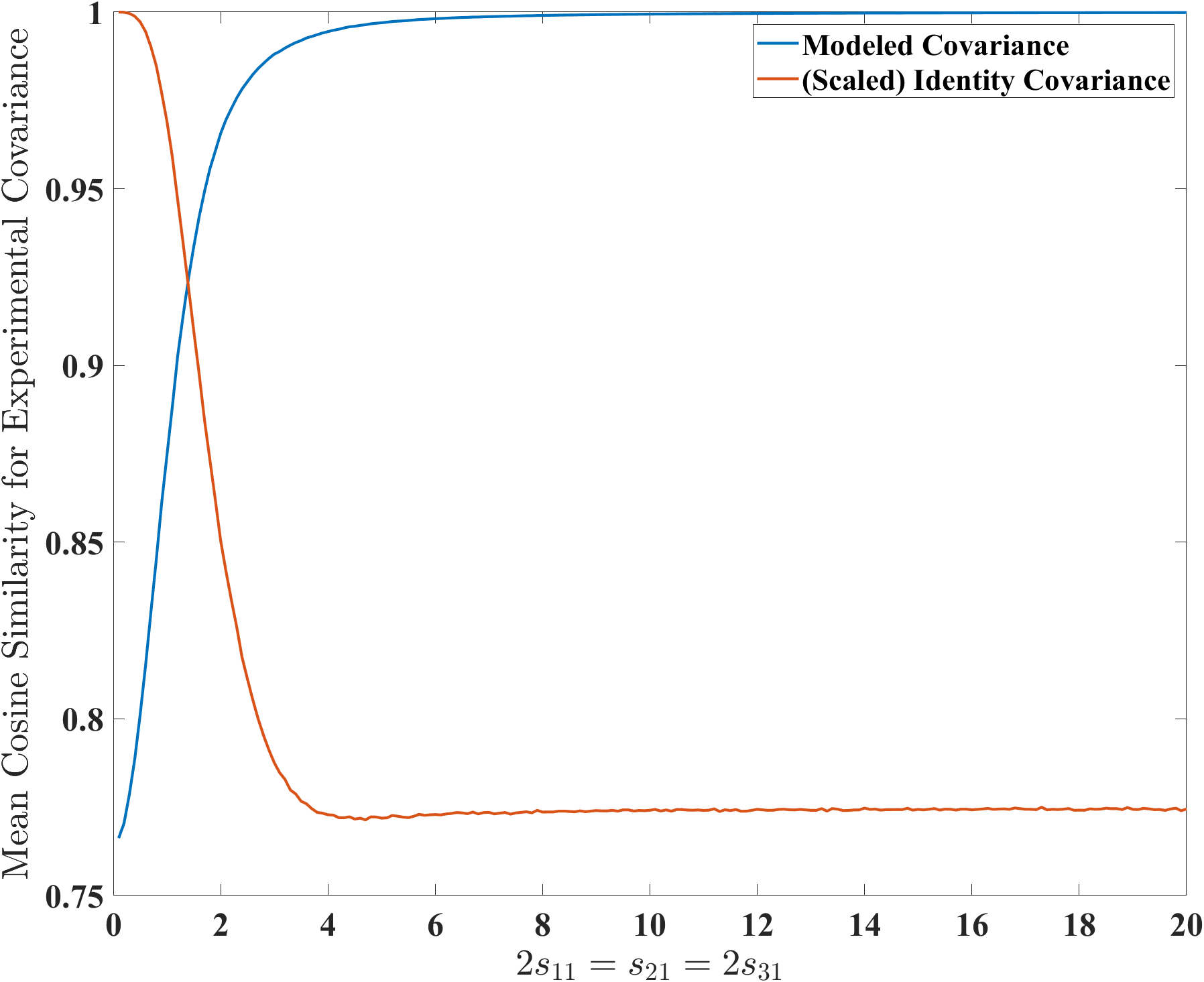}
    \caption{Mean cosine similarity for the modeled covariance matrix versus experimental covariance using $10^6$ trials at each $s_{21}$ value for single coil and $2s_{11} = s_{21} = 2s_{31}$.}
    \label{fig: verify approximation}
\end{figure}

From \eqref{eq: remainder notation}, the wrapped velocity measurements are linearly related to the phase differences; thus, the approximated scaled positive semi-definite covariance matrix for the noise, $\bm{n}$, in \eqref{eq: equality formulation} conditioned on no wrapping is
\begin{align}
    \label{eq: general covariance}
    \bm{\Sigma}\left(\bm{n}\right) \appropto  
    \frac{1}{\pi^2} \operatorname{diag} \left(\textbf{venc} \right) 
    \Pi \left(\bm{\Sigma} \left(\widetilde{\bm{\theta}}\right) \right)\operatorname{diag} \left( \textbf{venc} \right).
\end{align}

\subsection{Best Linear Unbiased Estimator}
\label{subsection: Linear Unbiased Estimator}
Based on the covariance derivation in \ref{subsection: Estimation of Phase Noise Covariance}, we can now formulate the estimator. 
We adopt two approximations suitable when the SNR is not extremely low.
Our first approximation is that $\bm{n}$ follows a joint Gaussian distribution with covariance matrix given in \eqref{eq: general covariance},
\begin{equation}
    \label{eq: assumption 1}
    \bm{n} \sim \mathcal{N} (0, \bm{\Sigma}(\bm{n})).
\end{equation}
Denote the velocity estimate $\widehat{v}$ given wrapping integers $\bm{k}$ as $\widehat{v}_{\bm{k}}$. Then, due to \eqref{eq: assumption 1}, $\widehat{v}_{\bm{k}}$ and its resulting RMSE given no wrapping, $\sigma(\widehat{v}_{\bm{k}})$, are 
\begin{align}
    \label{eq: weight combination}
    \widehat{v}_{\bm{k}}        &= \langle \bm{w}^\intercal (\widetilde{\bm{v}}+\bm{k}\circ 2\textbf{venc})\rangle_\Omega\\
    \label{eq: rmse given k}
    \sigma(\widehat{v}_{\bm{k}})&= \bm{w}^\intercal \bm{\Sigma}\left(\bm{n}\right) \bm{w},
\end{align}
where 
\begin{equation}
    \label{eq: w}
    \bm{w}^\intercal = \frac{\bm{1}^\intercal \bm{\Sigma}^{-1}\left(\bm{n}\right)}{\bm{1}^\intercal \bm{\Sigma}^{-1}\left(\bm{n}\right)\bm{1}},
\end{equation}
and $\langle \bm{a} \rangle_{\bm{b}}$ denotes remainders after elementwise modulo $\bm{a}$ by $\bm{b}$.
Thus, the estimate $\widehat{v}_{\bm{k}}$ is a weighted sum of unwrapped noisy velocities. 

The second assumption we adopt is 
\begin{equation}
    \label{eq: assumption 2}
    \mathbb{P} ( -\textbf{venc} \preccurlyeq \bm{n} \preccurlyeq \textbf{venc}) \approx 1,
\end{equation}
where $\preccurlyeq$ is element-wise less than or equal. This approximation 
yields the 
likelihood $\mathbb{P}(\widetilde{\bm{v}}|v)$
\begin{align}
    \label{eq: cost}
    \mathbb{P}(\widetilde{\bm{v}}|v)      &\approx \frac{e^{-\frac{1}{2}\operatorname{d}_{2\textbf{venc}}^\intercal(\widetilde{\bm{v}},v)\bm{\Sigma}^{-1}\left(\bm{n}\right)\operatorname{d}_{2\textbf{venc}}(\widetilde{\bm{v}},v )}}{\det (2\pi \bm{\Sigma}^{-1}\left(\bm{n}\right))},
\end{align}
where $\operatorname{d}_{\bm{z}}(\bm{x}, \bm{y})$ is a ``wrapped displacement'' between $\bm{x}$ and $\bm{y}$ with respect to $\bm{z}$, 
\begin{align}
    \label{eq: wrapped distance}
    \operatorname{d}_{\bm{z}} (\bm{x},\bm{y})            & \triangleq \bm{x}-\bm{y} - \left \lfloor (\bm{x}-\bm{y})\oslash \bm{z} \right \rceil\circ \bm{z}.
\end{align}
We use $-,\lfloor \cdot \rceil, \oslash$ to denote element-wise subtraction, rounding, and division.
One could perform search over all possible $\bm{k}$ to minimize the negative log likelihood, 
\begin{equation}
    \label{eq: negative log likelihood}
    \mathcal{L} \left(\widetilde{\bm{v}}, \widehat{v}_{\bm{k}}\right) = \tfrac{1}{2}\operatorname{d}_{2\textbf{venc}}^\intercal(\widetilde{\bm{v}},\widehat{v}_{\bm{k}})
    \bm{\Sigma}^{-1}\left(\bm{n}\right)\operatorname{d}_{2\textbf{venc}}(\widetilde{\bm{v}},\widehat{v}_{\bm{k}}).
\end{equation}
In \ref{subsection: PRoM}, we instead present a fast method to detect the best wrapping integers, $\bm{k}^\star$. In \ref{subsection: Three-point Encoding} and \ref{subsection: Design}, we explicitly consider three-point encoding, $N_e=3$, to provide concrete results and to optimize the design of \textbf{venc} and underlying first moments.

\subsection{PRoM}
\label{subsection: PRoM}
We introduce a fast estimator based on \eqref{eq: weight combination} and detection of the wrapping integers, $\bm{k}$. We refer to this estimator
as \emph{Phase Recovery from Multiple Wrapped Measurements} (PRoM). PRoM extends the prior signal processing result in \cite{Xia2015} to accommodate correlated phase errors and to provide a fast computation of the wrapping integers. Moreover, PRoM provides a grid-free alternative to grid search over $v$.

Assuming $v \in [0, \Omega)$, define the set $\mathcal{K}'(\widetilde{\bm{v}})$ of wrapping integers
\begin{align}
    \label{eq: K prime}
    \mathcal{K}'(\widetilde{\bm{v}}) &\triangleq \{\bm{k} |-\bm{1}  \preccurlyeq \bm{k} \preccurlyeq \bm{h}\}\\
    \label{eq: h}
    \bm{h} &= \Omega\oslash2\textbf{venc}.
\end{align}
By \eqref{eq: assumption 2}, $\mathbb{P} (\bm{k} \in \mathcal{K}'(\widetilde{\bm{v}})) \approx 1$. So, from \eqref{eq: weight combination} and \eqref{eq: cost}, we can minimize the negative log likelihood
\begin{equation}
\label{eq:  reduced negative log likelihood}
    \widehat{v} 
    =\underset{v \in [0,\Omega)}{\operatorname{argmin}}~ \mathcal{L}(\widetilde{\bm{v}}, v) 
    = \underset{v\in \{ \widehat{v}_{\bm{k}}| \bm{k} \in \mathcal{K}'(\widetilde{\bm{v}})\}}{\operatorname{argmin}}~ \mathcal{L}(\widetilde{\bm{v}}, v),
\end{equation}
with probability $\approx 1$. Next, we leverage the equality constraint in \eqref{eq: equality formulation} combined with the second approximation in \eqref{eq: assumption 2} to decrease the cardinality of $\mathcal{K}'(\widetilde{\bm{v}})$, denoted as $| \mathcal{K}'(\widetilde{\bm{v}})|$. We have
\begin{align*}
    \mathbb{P} \left( -\tfrac{1}{2} \preccurlyeq \bm{k} + (\widetilde{\bm{v}}-v)\oslash (2\textbf{venc}) \preccurlyeq \tfrac{1}{2}\right) \approx 1,
\end{align*}
which is equivalent to
\begin{align}
    \mathbb{P} \left( \bm{k} = \left\lceil -\tfrac{1}{2} - (\widetilde{\bm{v}}-v)\oslash (2\textbf{venc}) \right\rceil  \right) \approx 1,
\end{align}
where $\lceil \cdot \rceil$ is element-wise ceiling function. Then, the pruned search set for $\bm{k}$ can be expressed
\begin{align}
    \label{eq: K}
    \mathcal{K}(\widetilde{\bm{v}}) \triangleq \big\{\left\lceil -\tfrac{1}{2} - (\widetilde{\bm{v}}-v)\oslash (2\textbf{venc}) \right\rceil \big| v \in [0,\Omega) \big\}.
\end{align}
So, the parsimonious construction considers only 
$v \in [0,\Omega)$ such that $-\tfrac{1}{2} + \tfrac{v-\widetilde{v}_{ab}} {2\text{venc}_{ab}}$ is integer for any $2\text{venc}_{ab}$. The cardinality of the search set
\begin{equation*}
    |\mathcal{K}(\widetilde{\bm{v}})| \leq \bm{1}^\intercal\bm{h} .
\end{equation*}
Thus, the number of searches is bounded by the summation of $\bm{h}$ instead of product of $\bm{h}$. Then, the minimization in \eqref{eq:  reduced negative log likelihood} can be reduced to
\begin{equation}
    \label{eq: PRoM}
    \widehat{v} = \underset{v\in \{ \widehat{v}_{\bm{k}}| \bm{k} \in \mathcal{K}(\widetilde{\bm{v}})\}}{\operatorname{argmin}}~ \mathcal{L}(\widetilde{\bm{v}}, v),
\end{equation}
with probability $\approx 1$. Together, the construction of the pruned set, $\mathcal{K}(\widetilde{\bm{v}})$, of candidate wrapping integers and the efficient search over
$\{ \widehat{v}_{\bm{k}}| \bm{k} \in \mathcal{K}(\widetilde{\bm{v}})\}$ to minimize
$\mathcal{L}(\widetilde{\bm{v}}, v)$ comprise PRoM, an approximate MLE of $v$.
For a general $N_e$-point encoding, PRoM pseudo-code is given in Alg.~\ref{alg: PRoM}, and PRoM code is available at \url{https://github.com/Zhao-Shen/PRoM}.
\begin{algorithm}[!ht]
\caption{PRoM for $N_e$-point Encoding}
\label{alg: PRoM}
\begin{algorithmic}[1]
    \REQUIRE 
    Measurements, $\widetilde{\bm{Y}}$.
    First moments, $m_{11},...,m_{1N_e}$.
    \STATE Calculate $\textbf{venc}, \widetilde{\bm{v}}, \Omega,\mathcal{K}(\widetilde{\bm{v}})$ via (\ref{eq: vencs},~\ref{eq: congruence},~\ref{eq: Range Omega},~\ref{eq: K}).
    \STATE Calculate scaled $\bm{\Sigma}\left(\bm{n}\right)$ and $\bm{w}$ via (\ref{eq: general covariance},~\ref{eq: w}).
    \STATE Calculate $\widehat{v}$ via (\ref{eq: weight combination}, \ref{eq: PRoM}).
    \ENSURE $\widehat{v}$   
\end{algorithmic}
\end{algorithm}

Fig.~\ref{fig: cost} illustrates $\mathcal{L} (\widetilde{\bm{v}},\widehat{v})$ for $\textbf{venc} = [35, 10, 14]^\intercal$\,cm/s. The searched candidates $\{\widehat{v}_k| \bm{k} \in \mathcal{K}(\widetilde{\bm{v}})\}$ are marked by superimposed red dots. For this case, $|\mathcal{K}'(\widetilde{\bm{v}})|= 252$ and $ |\mathcal{K}(\widetilde{\bm{v}})|=14$;
thus, 94.4\% of the search space $\mathcal{K}'$ is bypassed via the proposed construction of $\mathcal{K}(\widetilde{\bm{v}})$.
If $\bm{n}$ is known to concentrate in a smaller volume compared to the second assumption \eqref{eq: assumption 2},  or $v$ is known and restricted to a range less than $\Omega$, then $\mathcal{K}(\widetilde{\bm{v}})$ may be further pruned accordingly. 

\begin{figure}
    \centering
    \includegraphics[width = \columnwidth]{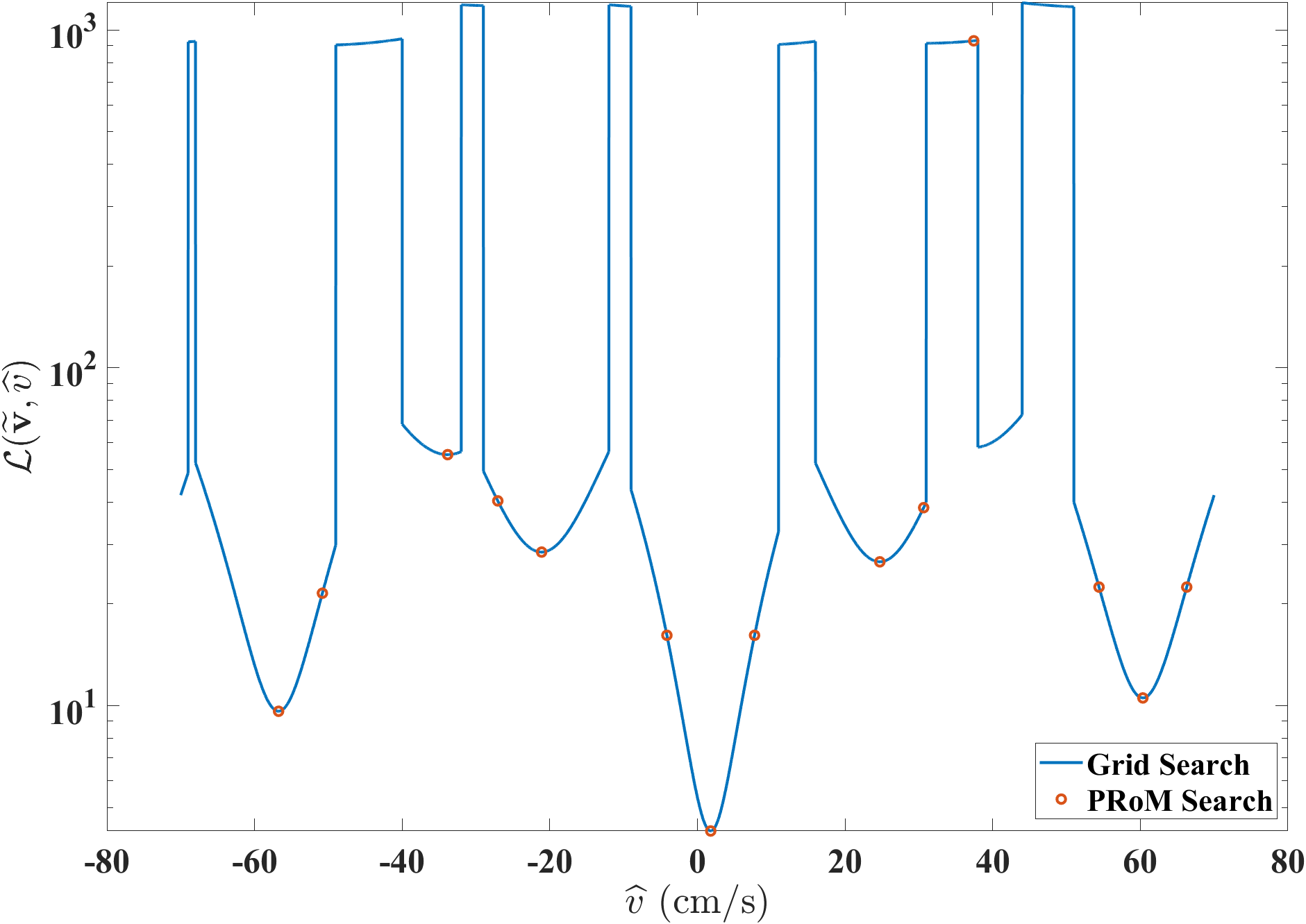}
    \caption{Negative log likelihood $\mathcal{L}(\widetilde{\bm{v}},\widehat{v}$), for single coil, symmetric three-point encoding, $\widetilde{\bm{v}} = [3,1,2]^\intercal$\,cm/s, $\textbf{venc} = [35,10,14]^\intercal$\,cm/s, $[s_{11} ,s_{21} , s_{31}] = [2.5,5,2.5]$. The red dots mark the PRoM searched candidates $\{\widehat{v}_k| \bm{k} \in \mathcal{K}(\widetilde{\bm{v}})\}$.}
    \label{fig: cost}
\end{figure}

To illustrate the reduction of computation complexity in PRoM, consider the case in Fig.~\ref{fig: cost}. Grid search MLE over velocity with spacing used in \cite{Carrillo2019} entails computation for $7000$ candidate velocities. In contrast, PRoM only requires search over only $|\mathcal{K}(\widetilde{\bm{v}}) |= 14$ candidates, yielding a $500$-times reduction.

PRoM admits a simple geometric interpretation. Observe that the noisy velocity measurement $\widetilde{\bm{v}}$ resides in a hyper-rectangle $\{\widetilde{\bm{v}}|\bm{0} \preccurlyeq \widetilde{\bm{v}} \preccurlyeq 2\textbf{venc}\}$.  The vector of noiseless velocity measurements $\langle v\rangle_{2\textbf{venc}}$ for $v \in [0,\Omega)$ is a point in the hyper-rectangle lying on wrapped line segments parallel to $\bm{1}$. Then, $\widehat{v}$ is found by an oblique projection of the noisy $\widetilde{\bm{v}} = \langle v + \bm{n}\rangle_{2\textbf{venc}}$ to the closest line segment. Here, the ``oblique projection" is determined by the $\bm{\Sigma}^{-1} (\bm{n})$ weighted distance. The search for the closest line segment is reduced to search for $\bm{k} \in \mathcal{K}(\widetilde{\bm{v}})$.

\subsection{Conditional Distribution of the Estimate}
\label{subsection: Conditional Distribution of the Estimate}
In this section, we derive $\mathbb{P}(\widehat{v}|v)$, the distribution of the estimated velocity given the true velocity; this somewhat technical derivation, in turn, enables the optimized design of \textbf{venc} and the underlying first moments. To derive the distribution, we first establish two lemmas. The first lemma tells us that adding the same constant, $\eta$, to all noise realization components does not affect the error in detecting the wrapping integers and simply adds $\eta$ to the PRoM velocity estimate, modulo $\Omega$.
\begin{lemma}
    \label{lemma 1}
    For $\eta \in \mathbb{R}$, let $\bm{n}' = \bm{n}+\eta $ and $ \widetilde{\bm{v}}' = \langle v  + \bm{n}'\rangle_{2\textbf{venc}}$. Then
    \begin{align}
        \widehat{v}(\widetilde{\bm{v}}') &= \langle \widehat{v}(\widetilde{\bm{v}}) +\eta \rangle_\Omega\\
        \langle\bm{k}^\star(\widetilde{\bm{v}})-\bm{k}(\widetilde{\bm{v}})\rangle_{\bm{h}} &= \langle\bm{k}^\star(\widetilde{\bm{v}}')-\bm{k}(\widetilde{\bm{v}}')\rangle_{\bm{h}} .
    \end{align}
\end{lemma}
\begin{proof}
By \eqref{eq: wrapped distance}, we have
    \begin{align*}
        \operatorname{d}_{2\textbf{venc}} (\widetilde{\bm{v}}', v+\eta) =& \operatorname{d}_{2\textbf{venc}} (\langle \widetilde{\bm{v}}+\eta \rangle_{2\textbf{venc}}, v+\eta)\\
        =& \operatorname{d}_{2\textbf{venc}} (\widetilde{\bm{v}}+\eta, v+\eta)\\
        =& \operatorname{d}_{2\textbf{venc}} (\widetilde{\bm{v}},v) .
    \end{align*}
Thus $\mathcal{L}(\widetilde{\bm{v}}, v) = \mathcal{L}(\widetilde{\bm{v}}', v+\eta)$ and $\widehat{v}(\widetilde{\bm{v}}') = \langle \widehat{v}(\widetilde{\bm{v}}) +\eta \rangle_\Omega$. Below we compare estimates using $\bm{k}^\star$ versus the true $\bm{k}$ along the $\bm{1}$ direction. By \eqref{eq: weight combination}
    \begin{align}
        &\langle\widehat{v}(\widetilde{\bm{v}}) - \bm{w}^\intercal(v+\bm{n}) \rangle_\Omega \nonumber\\
        \label{eq: v diff}
        =& \langle \bm{w}^\intercal(\langle \bm{k}^\star(\widetilde{\bm{v}})-\bm{k}(\widetilde{\bm{v}})\rangle_{\bm{h}}\circ 2\textbf{venc})\rangle_\Omega,
        \end{align}
    and
    \begin{align}
        &\langle\widehat{v}(\widetilde{\bm{v}}') - \bm{w}^\intercal(v+\bm{n}') \rangle_\Omega \nonumber\\
        \label{eq: v' diff}
        =& \langle \bm{w}^\intercal(\langle \bm{k}^\star(\widetilde{\bm{v}}')-\bm{k}(\widetilde{\bm{v}}')\rangle_{\bm{h}}\circ 2\textbf{venc})\rangle_\Omega.
    \end{align}
By the previously derived $\widehat{v}(\widetilde{\bm{v}}') = \langle \widehat{v}(\widetilde{\bm{v}}) +\eta \rangle_\Omega$, so we have
    \begin{align*}
        & \langle\widehat{v}(\widetilde{\bm{v}}') - \bm{w}^\intercal(v+\bm{n}') \rangle_\Omega\\
        =& \langle\widehat{v}(\widetilde{\bm{v}})+\eta - \bm{w}^\intercal(v+\bm{n}) - \eta \rangle_\Omega\\
        =& \langle\widehat{v}(\widetilde{\bm{v}}) - \bm{w}^\intercal(v+\bm{n}) \rangle_\Omega .
    \end{align*}
    Thus \eqref{eq: v diff} and \eqref{eq: v' diff} are equal for all $\bm{w}$, and
    \begin{equation*}
        \langle\bm{k}^\star(\widetilde{\bm{v}})-\bm{k}(\widetilde{\bm{v}})\rangle_{\bm{h}} = \langle\bm{k}^\star(\widetilde{\bm{v}}')-\bm{k}(\widetilde{\bm{v}}')\rangle_{\bm{h}} .
    \end{equation*}
\end{proof}
Fig.~\ref{fig: distribution of k star} provides visualisation of wrapping integers $\bm{k}^\star(\widetilde{\bm{v}})$ for the case $\textbf{venc} = [35,10,14]^\intercal, [s_{11}, s_{21}, s_{31}] =[2.5, 5, 2.5]$. All $\widetilde{\bm{v}}$ along direction $\bm{1}$ inside the hyper-rectangle share the same $\bm{k}^\star$, which is a consequence of Lemma \ref{lemma 1}.
\begin{figure}[h!]
    \centering
    \includegraphics[width = \columnwidth]{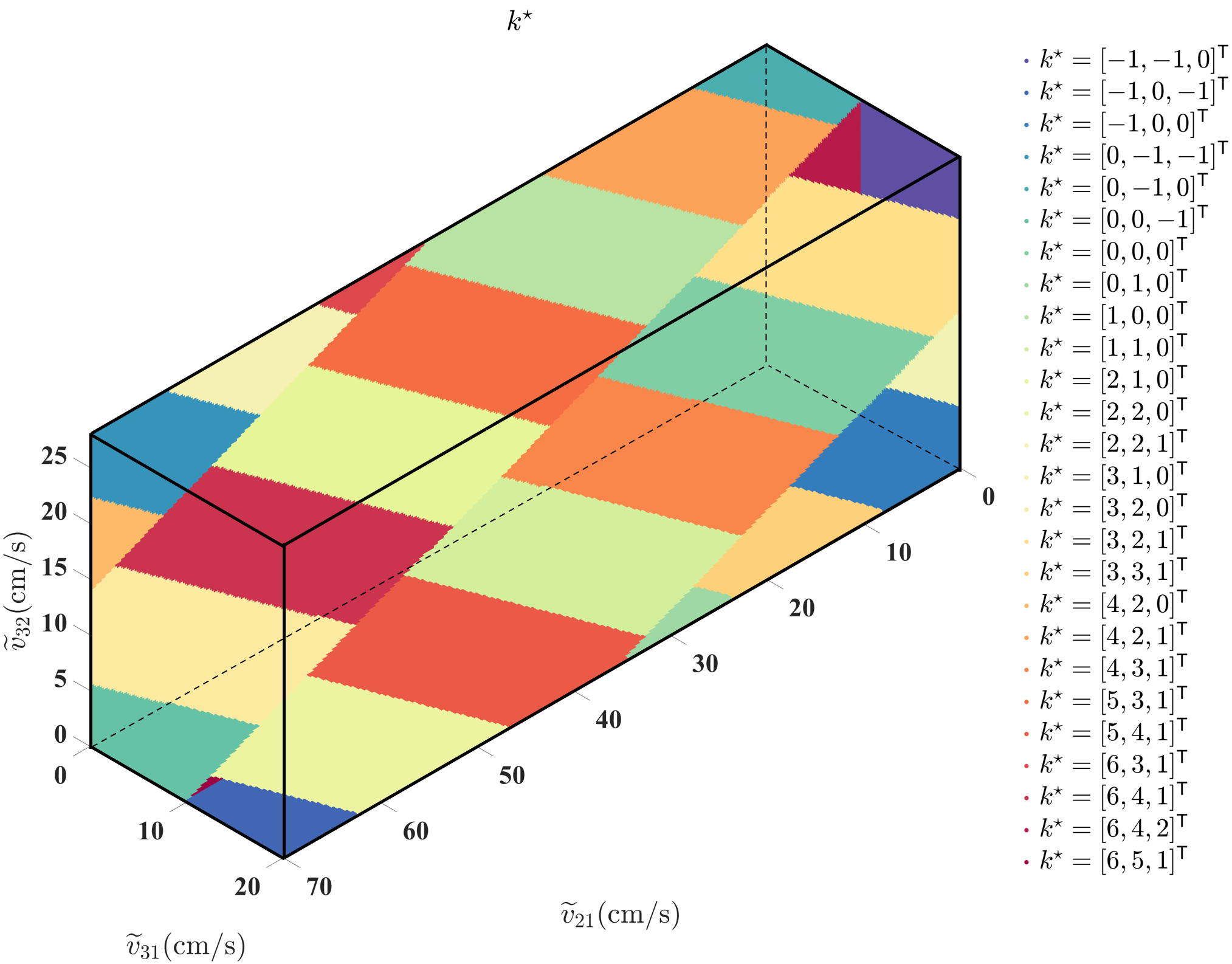}
    \caption{Illustration of wrapping integers $\bm{k}(\widetilde{\bm{v}})$, $N_c=1$, $N_e = 3$, $\textbf{venc} = [35,10,14]^\intercal$, $[s_{11} ,s_{21} , s_{31}] = [2.5,5,2.5]$. The colored tiles for each region are drawn on the surfaces of the hyper-rectangle, and regions extend unchanged parallel to the $[1,1,1]^\intercal$ direction, which is the direction of the line of sight in the figure. PRoM detects wrapping integers for fast, accurate estimation of velocity.}
    \label{fig: distribution of k star}
\end{figure}

In addition, $\mathcal{L}(\widetilde{\bm{v}},v_{\bm{k}})$ as a function of $\widetilde{\bm{v}}$ is piece-wise quadratic with the same curvature for all $\bm{k}$, so the decision boundaries of $\bm{k}^\star (\widetilde{\bm{v}})$ are linear, which is also illustrated in Fig.~\ref{fig: distribution of k star}.

By Lemma \ref{lemma 1},  the error in wrapping integers, $\langle\bm{k}^\star(\widetilde{\bm{v}})-\bm{k}(\widetilde{\bm{v}})\rangle_{\bm{h}}$, remains constant for all noise realizations $\bm{n}$ along any line parallel to $\bm{1}$.
So, we can divide the space of all possible noise realizations, $\mathbb{R}^{\frac{N_e(N_e-1)}{2}}$, into ``tubes'' $\mathcal{T}(\bm{x})$ parallel to $\bm{1}$, based on the difference, $\bm{x}$, of estimated wrapping integers $\bm{k}^\star$ and true wrapping integers $\bm{k}$.
\begin{equation}
    \label{eq: tube}
    \mathcal{T}(\bm{x}) \triangleq \{ \bm{n}| \langle \bm{k}^\star (\widetilde{\bm{v}})-\bm{k}(\widetilde{\bm{v}})\rangle_{\bm{h}} = \bm{x} \}.
\end{equation}
We can, as seen below, integrate over these tubes to arrive at error probabilities for detecting the wrapping integers.
Next we establish a second lemma describing the orthogonality between pairwise noise differences and the error in the estimated velocity. 
\begin{lemma}
    \label{lemma 2}
    \begin{equation}
        n_{ab}-n_{cd} \perp \bm{w}^\intercal \bm{n}, ab \not = cd
    \end{equation}
\end{lemma}
\begin{proof}
    \begin{align*}
        &\text{cov}(\bm{w}^\intercal \bm{n}, n_{ab}-n_{cd})\\
        =& \bm{w}^\intercal \mathbb{E} (\bm{n}  (n_{ab}-n_{cd}))\\
        =& \frac{\bm{1}^\intercal \bm{\Sigma}^{-1}\left(\bm{n}\right)\bm{\Sigma}\left(\bm{n}\right)(\bm{e}_{ab}-\bm{e}_{cd})}{\bm{1}^\intercal \bm{\Sigma}\left(\bm{n}\right)^{-1}\bm{1}} = 0,
    \end{align*}
    where $\bm{e}_{ab}$ is the standard basis equal to $1$ at one position corresponding to $n_{ab}$ in $\bm{n}$, and $0$ otherwise.
\end{proof}
Armed with the two lemmas, we can specify the distribution.
\begin{theorem}
    \label{theorem 1}
    Given $v$, $\forall \bm{n} \in \mathcal{T}(\bm{x})$, $\widehat{v}(\widetilde{\bm{v}})$ follows wrapped normal distribution $\mathcal{N}(\langle v + \bm{w}^\intercal (\bm{x}\circ 2\textbf{venc})\rangle_\Omega, \bm{w}^\intercal \bm{\Sigma}\left(\bm{n}\right)\bm{w})$.
\end{theorem}
\begin{proof}
Note that the event $\bm{n} \in \mathcal{T}(\bm{x})$ is only determined by the pairwise difference of $n_{ab}-n_{cd}$. Thus, by Lemma 2 and the Gaussian assumption,  the random variable
$\bm{w}^\intercal \bm{n}$ is independent of the membership $\bm{n} \in \mathcal{T}(\bm{x})$.
Hence, for every region $\mathcal{T}(\bm{x})$, the conditional distribution of velocity estimate from noisy data is given by
    \begin{align*}
        \widehat{v}(\widetilde{\bm{v}}) |  \bm{n} \in \mathcal{T}(\bm{x}) 
        &= \langle \bm{w}^\intercal (v+\bm{x}\circ 2\textbf{venc} + \bm{n})\rangle_\Omega \\
        &= \langle v + \bm{w}^\intercal (\bm{x}\circ 2\textbf{venc}) + \bm{w}^\intercal \bm{n}\rangle_\Omega. 
    \end{align*}
\end{proof}
Let $f(\widehat{v}| \mathcal{T}(\bm{x}), v)$ denote the conditional Gaussian probability density function (pdf) in Theorem~\ref{theorem 1} without wrapping by modulo $\Omega$. Thus, by wrapping the pdf function and invoking the law of total probability, we have
\begin{align}
\label{eq: vhatdist}
    \mathbb{P}(\widehat{v}|v) &= \sum_{\bm{x}} \mathbb{P}(\bm{n}\in \mathcal{T}(\bm{x}) | v) f_{\Omega}(\widehat{v}| \mathcal{T}(\bm{x}), v)\\
    f_{\Omega}(\widehat{v}| \mathcal{T}(\bm{x}), v) &= 
    \begin{cases}
        \sum_{l\in \mathbb{Z}} f(\widehat{v}+l\Omega| \mathcal{T}(\bm{x}), v),& \widehat{v} \in [0,\Omega)\\
        0,& \text{otherwise}
    \end{cases}. \nonumber
\end{align}
To complete \eqref{eq: vhatdist}, we need only to calculate $\mathbb{P}(\bm{n}\in \mathcal{T}(\bm{x}) | v)$ by integration of a multivariate normal distribution. Leveraging Lemma \ref{lemma 1}, the integration can be simplified to a definite integration of a normal distribution in $\frac{N_e(N_e-1)}{2}-1$ variables.

Fig.~\ref{fig: distribution} illustrates the histogram of $\widehat{v}|v=0$ using $10^5$ trials, $N_e = 3$, $N_c = 1$, $\textbf{venc} = [99,18,22]^\intercal$\,cm/s at $s_{21} =5,10$. Invoking Theorem \ref{theorem 1}, we predict the five most probable wrapping integers, which correspond to detection regions $\mathcal{T}([0,0,0]^\intercal), \mathcal{T}([5,4,1]^\intercal), \mathcal{T}([6,5,1]^\intercal),\mathcal{T}([1,1,0]^\intercal)$, and $\mathcal{T}([10,8,2]^\intercal)$.
These five predicted detections correspond to $f(\widehat{v}| \mathcal{T}(\bm{x}), v)$ centered at $0$ (the true velocity), $\pm 177.81$, and $\pm 40.37$\,cm/s; these five predictions are validated in the histogram. For the higher SNR case of 
$s_{21} = 10$, the probability of unwrapping errors goes very small, and the $10^5$ trials are insufficient
to encounter an unwrapping error to $\pm 40.37$\,cm/s.

\begin{figure}[h!]
    \centering
    \includegraphics[width = \columnwidth]{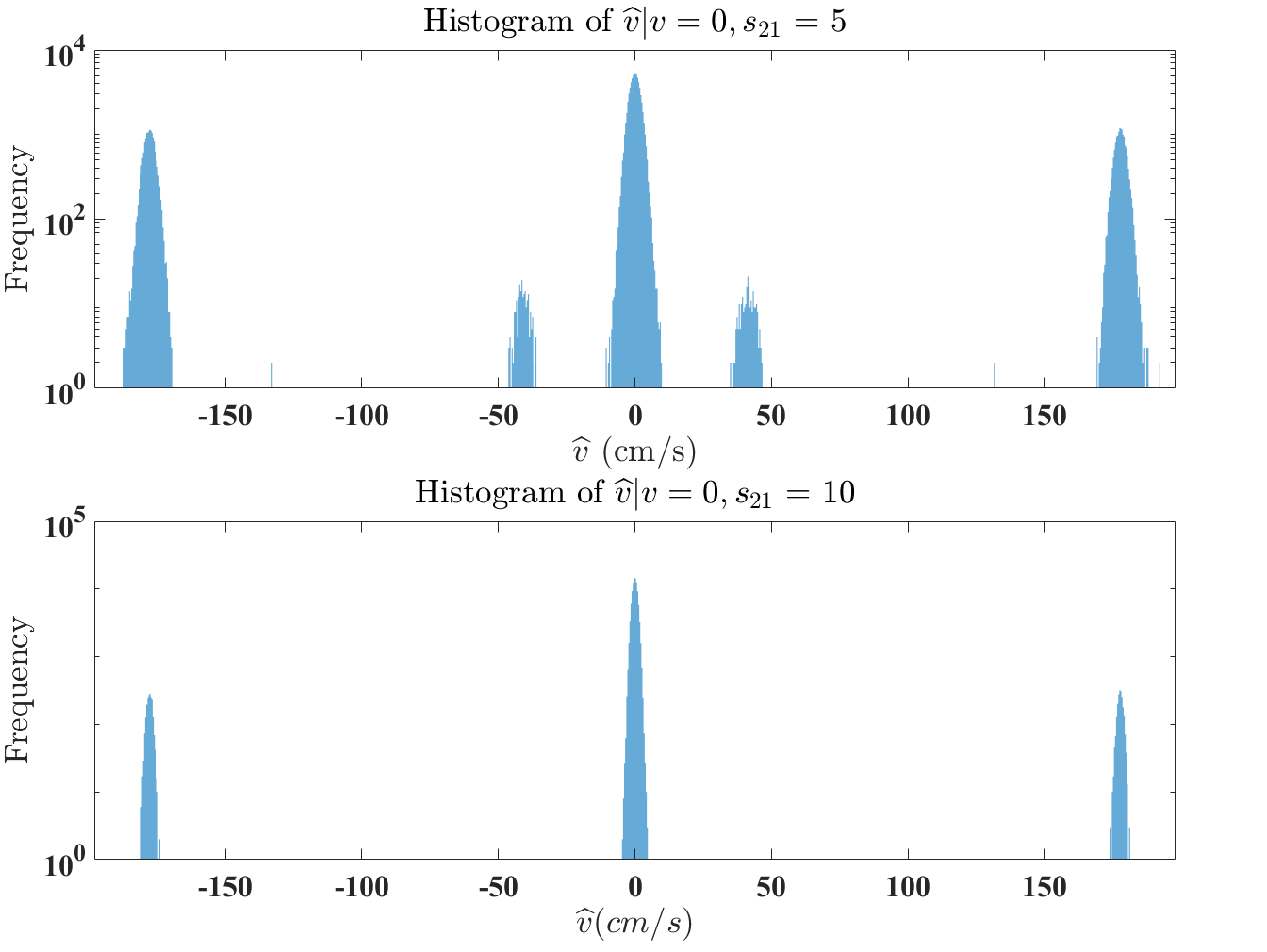}
    \caption{Histogram from $10^5$ trials of $\widehat{v}|v=0$, for $N_c=1$, $N_e=3$, $\textbf{venc} = [99,18,22]^\intercal$\,cm/s, $v = 0$, at $s_{21} = 5$ (top) and $s_{21} = 10$ (bottom). $2s_{11} = s_{21} = 2s_{31}$. Note the logarithmic scale for the vertical axis covering four (top) and five (bottom) orders of magnitude. The histogram illustrates both noise sensitivity via the spread of each Gaussian component and the probability of unwrapping errors via the presence of multiple components.}
    \label{fig: distribution}
\end{figure}

\subsection{Three-point Encoding}
\label{subsection: Three-point Encoding}
We consider here three-point encoding for velocity in one direction. Due to (\ref{eq: vencs}, \ref{eq: Range Omega}), every $\text{venc}_{ab}$, and hence the unambiguous range, $\Omega$, depends only on the differences between first moments; thus, $\text{venc}_{ab}$ and $\Omega$ are unaffected by adding the same constant to every first moment. We assume the following ordering for the first moments:
\begin{equation}
    \label{eq: symm3pt}
    m_{11} < m_{12}  < m_{13},\,  m_{12}-m_{11} < m_{13}-m_{11}.
\end{equation}
Thus, the three vencs are ordered: $\text{venc}_{31}<\text{venc}_{32}<\text{venc}_{21}$. Let $\xi = \frac{\text{venc}_{32}}{\text{venc}_{31}}$, noting that \eqref{eq: vencs} and \eqref{eq: symm3pt} imply $\xi \in (1,2)$. For rational $\xi=p/q$ with co-prime integers $p$ and $q$, the unambiguous range $\Omega$ in \eqref{eq: Range Omega} is 
\begin{equation}
    \label{eq: range vs largest venc}
    \Omega = 2(p-q)\text{venc}_{21} .
\end{equation}
Thus, by jointly unwrapping multiple vencs, one can construct an unaliased velocity range that is larger than the highest venc, $\text{venc}_{21}$, by a factor of $2(p-q)$. The covariance matrix for three-point encoding can be formulated via \eqref{eq: general covariance}. Correspondingly, we can calculate the combination weights $\bm{w}$ and the resulting RMSE given wrapping integers (\ref{eq: weight combination}, \ref{eq: rmse given k}). Armed with the explicit error variance and the probability of unwrapping errors derived below, we present in \ref{subsection: Design} an optimized design of \textbf{venc} for three-point encoding with the constraint on the largest first moment, given a desired unambiguous range of velocities to be observed and SNR level for each encoding.

\subsection{Existing Estimators for Three-point Encoding}
\label{subsection: Existing Dual-Venc Estimators}
In the notation of \eqref{eq: remainder notation} and \eqref{eq: definition}, the approaches in \cite{Lee1995}, \cite{Schnell2017} use the unaliased high venc measurement, $\widetilde{v}_{21}$, to unwrap the low venc measurement, $\widetilde{v}_{31}$, while $\text{venc}_{32}$ goes unused. The estimator in \cite{Schnell2017}, which we denote as \emph{standard dual-venc} (SDV),
is given by
\begin{equation}
    \widehat{v} = 
    \begin{cases}
        \widetilde{v}_{21} - 4\text{venc}_{31},~ \tfrac{\widetilde{v}_{21}-\widetilde{v}_{31}}{2\text{venc}_{31}} \in (-2.4, -1.6)\\
        \widetilde{v}_{21} - 2\text{venc}_{31},~ \tfrac{\widetilde{v}_{21}-\widetilde{v}_{31}}{2\text{venc}_{31}} \in (-1.2, -0.8)\\
        \widetilde{v}_{21} + 2\text{venc}_{31},~ \tfrac{\widetilde{v}_{21}-\widetilde{v}_{31}}{2\text{venc}_{31}} \in (0.8, 1.2) \\
        \widetilde{v}_{21} + 4\text{venc}_{31},~ \tfrac{\widetilde{v}_{21}-\widetilde{v}_{31}}{2\text{venc}_{31}} \in (1.6, 2.4) .
    \end{cases}
\end{equation}
In \cite{Carrillo2019}, two potentially aliased measurements $\widetilde{v}_{31}, \widetilde{v}_{32}$ are jointly unwrapped by minimizing
\begin{align}
    \label{eq: ODV cost}
    \widehat{v} = \underset{{v\in [-\Omega/2,\Omega/2)}}{\operatorname{argmin}}~
    \sum_{l=(31,32)} \left(1- \cos \left( \tfrac{\pi v}{ \text{venc}_l}  - \widetilde{\theta}_l \right)\right).
\end{align}
The authors dub their approach  \emph{optimal dual-venc} (ODV), and the minimization is accomplished by searching $v\in [-\Omega/2,\Omega/2)$ with grid spacing $\frac{\text{venc}_{31}}{1000}$. The cost adopted in ODV is equivalent to
\begin{equation}
\label{eq: ODV2}
    \tfrac{1}{2} \left|e^{i\frac{\pi v}{\text{venc}_{31}}}-e^{i\widetilde{\theta}_{31}} \right|^2 + \tfrac{1}{2} \left|e^{i\frac{\pi v}{\text{venc}_{32}}}-e^{i\widetilde{\theta}_{32}} \right|^2,
\end{equation}
which intrinsically assumes no correlation between the noisy phase differences, $\widetilde{\theta}_{31}$ and $\widetilde{\theta}_{32}$. 
The ODV approach recommends $\text{venc}_{32} = \frac{3}{2} \text{venc}_{31}$, yielding an unambiguous velocity range of length $2 \text{venc}_{21}$, which is twice the highest venc.
The choice $\frac{3}{2}$ is a heuristic to lessen the probability of unwrapping errors when minimizing \eqref{eq: ODV cost} in the presence of noise. 

The non-convex optimization (NCO) in \cite{loecher2018velocity} iteratively minimizes a cost similar to \eqref{eq: ODV2} with weights to accommodate a lower SNR in presence of intra-voxel dephasing:
\begin{align}
    |\widetilde{r}_{31} |^2 \left|e^{i\frac{\pi v}{\text{venc}_{31}}}-e^{i\widetilde{\theta}_{31}} \right|^2 + 
    |\widetilde{r}_{32} |^2 \left|e^{i\frac{\pi v}{\text{venc}_{32}}}-e^{i\widetilde{\theta}_{32}} \right|^2.
\end{align}
Both ODV and NCO can be applied to any number of encodings; further, the NCO algorithm also incorporates a spatial regularization across voxels in the form of the Laplacian of the velocity map.

\subsection{Design for Three-point Encoding}
\label{subsection: Design}
The selection of the pair $\left\{ \text{venc}_{31}, \xi \right\}$ defines first moments $[m_{11},m_{12},m_{13}]$ up to translation. To minimize the worst intra-voxel dephasing, we choose symmetric encoding $m_{11} = -m_{13}$; to further ameliorate intra-voxel dephasing, we allow for a user-defined upper bound on the largest first moment,
$m_{13} \leq m_\tau$. The choice of \textbf{venc} entails the interplay of noise sensitivity, probability of correct unwrapping, and the range of reliably unaliased velocities. We adopt a performance guarantee strategy for navigating these competing objectives.  The design inputs are:
the maximum range of velocities to be reliably detected, $\Omega_{\bm{\epsilon}}$;
a lower bound of operating measurement SNR, $s_{\alpha \beta}$;
an upper bound, $m_\tau$, on the magnitude of the largest first moment; 
and, bounds on the per-voxel probability of an unwrapping error. The design outputs are the \textbf{venc} and an underlying $[m_{11},m_{12},m_{13}]$. To formalize the notion of optimality, we make four definitions.

\begin{definition}[Unwrapping error]
    If $\langle \bm{k}^\star (\widetilde{\bm{v}}) -\bm{k}(\widetilde{\bm{v}}) \rangle_{\bm{h}} \not = \bm{0}$, i.e., wrapping integers are incorrectly detected, then we say an \textit{Unwrapping Error} occurs.
\end{definition}

\begin{definition}[Aliasing error]
    Given no unwrapping error, if $\widehat{v}_{\bm{k}^\star}$ is aliased, then we say an \textit{Aliasing Error} occurs.
\end{definition} 

\begin{definition}[$\bm{\epsilon}$-Reliable Range]
    For given measurement SNR, $s_{\alpha\beta}$, and vector of two small numbers $\bm{\epsilon} = [\epsilon_1, \epsilon_2]^\intercal$, the $\bm{\epsilon}$-\textit{reliable range} $\Omega_{\bm{\epsilon}}$ is the range of the velocities for which $\mathbb{P}(\text{Unwrapping Error}) \leq \epsilon_1$ and $\mathbb{P}(\text{Aliasing Error}) \leq \epsilon_2$.
\end{definition}

Definition 3 allows us to specify a reliable velocity range $\Omega_{\bm{\epsilon}} <  \Omega $ to guard against aliasing. Armed with these three definitions, we can now state a precise meaning of optimality for three-point design.
\begin{definition}[Optimal \textbf{venc} Design] 
\label{def:opt}
    Given the SNR for the complex-valued data $s_{\alpha \beta}$, the desired maximum velocity range of length $\Omega_{\bm{\epsilon}}$, an upper bound $m_\tau$ on the largest first moment, and unwrapping error bounds  $\bm{\epsilon}$, the optimal \textbf{venc} minimizes the RMSE among all designs for which the unwrapping and aliasing errors satisfy $\mathbb{P}( \text{Unwrapping Error}) \leq \epsilon_1$  and $\mathbb{P}( \text{Aliasing Error}) \leq \epsilon_2$ across the entire range of length $\Omega_{\bm{\epsilon}}$.
\end{definition}

Given SNR and \textbf{venc}, $\mathbb{P}( \text{Unwrapping Error})$ can be calculated through Monte Carlo simulation by setting $v=0$ and counting the trials for which $\langle \bm{k}^\star (\widetilde{\bm{v}}) -\bm{k}(\widetilde{\bm{v}}) \rangle_{\bm{h}} \not = \bm{0}$. Independence of $\mathcal{T}(\bm{x})$ and  $\widehat{v}(\tilde{\bm{v}})$
allows simulation of $v=0$ to be sufficient.  The number of trials is selected as $100/\epsilon_1$ \cite{jeruchim1984techniques}. 

Bounding the Aliasing Error can be achieved via explicitly designing $\Omega$ different from $\Omega_{\bm{\epsilon}}$. Let the normal cumulative distribution function be $\Phi(\cdot)$. Then, $\mathbb{P}( \text{Aliasing Error}) \leq \epsilon_2$ when
\begin{equation}
    \Omega - \Omega_{\bm{\epsilon}} \geq 2\Phi^{-1}(1-\epsilon_2) \sqrt{\bm{w}^\intercal \bm{\Sigma}\left(\bm{n}\right)\bm{w}}    .
\end{equation}

The design procedure in Alg.~\ref{alg: PRoM Optimal Design for Three-Point Encoding} is an offline finite search to optimally design \textbf{venc} and the corresponding first moments.
To explore design options for $\xi \in (1,2)$, we search rational values $\xi = p/q$ among
\begin{equation}
   \left\{\tfrac{p}{q}\big| \gcd(p,q) = 1, \tfrac{p}{q} \in (1,2), p \leq P, q\leq Q \right\},
\end{equation}
where $\gcd(\cdot,\cdot)$ is greatest common divisor, and $P,Q \in \mathbb{Z}^+$ are predefined upper bounds on the positive integers $p,q$.

\begin{algorithm}[!ht]
\caption{{PRoM} Optimal Design for Three-Point Encoding}
\label{alg: PRoM Optimal Design for Three-Point Encoding}
\begin{algorithmic}[1]
    \REQUIRE 
    $P, Q, s_{\alpha\beta}, \bm{\epsilon}, \Omega_{\bm{\epsilon}}$, $m_\tau$
    \STATE $\sigma \leftarrow \infty$, $p \leftarrow 2$
    \WHILE{$p \leq P$}
    \STATE $q \leftarrow 1$
    \WHILE{$q \leq Q$}
    \IF{$\gcd(p,q) =1$ and $\frac{p}{q} \in (1,2)$}
    \STATE $v \leftarrow 0$, $\textbf{venc} \leftarrow [pq, q(p-q), p(p-q)]$.
    \STATE Simulate $\widetilde{\bm{v}}$ with $100 \epsilon_1^{-1}$ trials, apply PRoM.
    \IF{(Frequency of $\langle \bm{k}^\star (\widetilde{\bm{v}}) -\bm{k}(\widetilde{\bm{v}}) \rangle_{\bm{h}} \not = \bm{0}) < \epsilon_1$}
    \STATE $c \leftarrow \frac{\Omega_{\bm{\epsilon}}}{\Omega- 2\Phi^{-1}(1-\epsilon_2) \sqrt{\bm{w}^\intercal \bm{\Sigma}\left(\bm{n}\right)\bm{w}}}$
    \STATE $c \leftarrow \max \left[c, \frac{\pi}{2\gamma m_\tau q(p-q)} \right]^\intercal$
    \IF {$c\sqrt{\bm{w}^\intercal \bm{\Sigma}\left(\bm{n}\right)\bm{w}} \leq \sigma$}
    \STATE $\textbf{venc} \leftarrow c\textbf{venc}$, $\sigma \leftarrow c\sqrt{\bm{w}^\intercal \bm{\Sigma}\left(\bm{n}\right)\bm{w}}$
    \ENDIF
    \ENDIF
    \ENDIF
    \STATE $q \leftarrow q+1$
    \ENDWHILE
    \STATE $p \leftarrow p+1$
    \ENDWHILE
    \STATE $m_{11} \leftarrow \frac{-\pi}{2\gamma \text{venc}_{31}}, m_{12} \leftarrow \frac{\pi}{2\gamma \text{venc}_{31}} - \frac{\pi}{\gamma \text{venc}_{32}}, m_{13} \leftarrow -m_{11}$
    \ENSURE \textbf{venc}, $[m_{11},m_{12},m_{13}]$   
\end{algorithmic}
\end{algorithm}
The output $[m_{11},m_{12},m_{13}]$ of Alg.~\ref{alg: PRoM Optimal Design for Three-Point Encoding} is a symmetric encoding that can be translated by $\pm m_{13}$ to yield a referenced encoding; however, the referenced encoding suffers an increased risk of severe intra-voxel dephasing, owing to the doubling of the largest first moment.

\subsection{Post-processing using Spatial Information: PRoM+}
In this subsection, we propose a simple but effective post-processing strategy paired with PRoM. The PRoM per-voxel estimation can benefit from leveraging spatial correlations among the per-voxel phase unwrapping integers. We assume that the noiseless velocity map $u(\bm{\rho})$ belongs to a surface class $\mathcal{U}$, where $\bm{\rho}$ denotes spatial position. For example, a polynomial model has been used for the brain image phase \cite{liang1996model}, and the Hagen–Poiseuille equation has been used for laminar blood flow throughout most of the circulatory system \cite{sutera1993history}.

When the model is accurate, the difference between the noisy unbiased estimated and true velocity map at each location should be at the noise level, and we assume at each location the difference follows i.i.d.\ normal distribution with variance $\frac{1}{2\lambda}$. 

Using \eqref{eq: negative log likelihood}, the spatial post-processing can be expressed as minimizing the negative log likelihood 
\begin{equation}
    \label{eq: joint processing using spatial information}
    \underset{u \in \mathcal{U}, \widehat{v}(\bm{\rho}) \in \{\widehat{v}_{\bm{k}}(\bm{\rho})\}}{\operatorname{argmin}} ~  \sum_{\bm{\rho}} \mathcal{L}(\widetilde{\bm{v}}(\bm{\rho}), \widehat{v}(\bm{\rho})) + \lambda (\widehat{v}(\bm{\rho}) - u(\bm{\rho}))^2
\end{equation}
Here, to avoid over-smoothness due to the regularization using $u \in \mathcal{U}$, we restrict $\widehat{v}(\bm{\rho}) \in \{\widehat{v}_{\bm{k}}|\bm{k} \in \mathcal{K}\left(\widetilde{\bm{v}}(\bm{\rho}) \right) \}$, i.e., we only allow spatial information to affect $\bm{k}$.

To optimize \eqref{eq: joint processing using spatial information}, we adopt an alternating minimization strategy. For current $\widehat{v}(\bm{\rho})$, we fit it with the best $u(\bm{\rho}) \in \mathcal{U}$ via surface fitting. For current $u (\bm{\rho})$, we update the choice of $\widehat{v}_{\bm{k}}(\bm{\rho})$ per voxel to minimize the cost. These two steps guarantee convergence in terms of the cost function. Iterations continue to convergence, which for the integer-valued $\bm{k}$ simply means no change; no convergence threshold is required, as would be with real-valued variables. Convergence is observed in two iterations in all experiments reported below. To reduce computation, we consider only a few most likely velocity candidates $\widehat{v}_{\bm{k}}(\bm{\rho})$. Moreover, we mask out air regions through magnitude thresholding to reduce computation.

The PRoM estimator, together with the spatial post-processing, is denoted ``PRoM+''. In the section below, we adopt for $\mathcal{U}$ a basic non-parametric local quadratic regression for both phantom and in vivo experiments.

\section{Methods}
\label{section: Methods}

\subsection{Simulation}
\label{subsection: Simulation}
To validate the two assumptions (\ref{eq: assumption 1}, \ref{eq: assumption 2}) used in PRoM,  we compare the RMSE for PRoM, the RMSE for grid search MLE, and the square root of the Cram\'{e}r-Rao lower bound (CRLB) \cite[p.~364]{StoicaMoses} derived from complex measurements in \eqref{eq: datamodel}.
Results are computed for $\textbf{venc} = [15, 6, 10]^\intercal$\,cm/s and $50\%$ intra-voxel dephasing of amplitudes for high first moments: $2s_{11} = s_{21} = 2s_{31}$. RMSE values for both the MLE from the complex measurements and PRoM are each calculated using $10^5$ trials, where the grid search of MLE on $v$ has spacing $0.006$\,cm/s to reduce bias from gridding.

To compare the performance of PRoM versus SDV and ODV, we process the same measurements using different estimators. Simulation parameters include: $\textbf{venc} = [15, 6, 10]^\intercal$\,cm/s, $[s_{11}, s_{21}, s_{31}] =[10, 20, 10]$, $N_c = 1$ coil. The ODV estimation algorithm uses $\widetilde{v}_{31}, \widetilde{v}_{32}$, while SDV uses $\widetilde{v}_{31},\widetilde{v}_{21}$. We calculate RMSE averaged over $10^5$ trials at each true $v$ on $[-30, 30]$\,cm/s sampled every $0.1$\,cm/s. 

To assess the optimized encoding design in Alg.~\ref{alg: PRoM Optimal Design for Three-Point Encoding},  we set a required velocity range of $[-150, 150]$\,cm/s and compare suggested choices of symmetric three-point encoding for each algorithm. Simulation parameters include: $N_c = 1$ coil, $s_{21} = 20$.  SDV suggest using $\textbf{venc} = [150, 60, 100]^\intercal$\,cm/s, ODV recommends $\xi = \frac{3}{2}$ and specify the three vencs to be $\textbf{venc} = [150, 50,75]^\intercal$\,cm/s. For PRoM, we assume intra-voxel dephasing leads to $[s_{11}, s_{21},s_{31}] = [10, 20, 10]$, other input includes $[P, Q] = [10,10], \bm{\epsilon} = [10^{-7}, 10^{-7}]^\intercal, \Omega_{\bm{\epsilon}} = 300$\,cm/s, $\gamma m_\tau = \frac{\pi}{50}$\,s/cm. The design procedure in Alg.~\ref{alg: PRoM Optimal Design for Three-Point Encoding} gives $\textbf{venc} = c [30,5,6]^\intercal$\,cm/s for scaling $c=5.1242$, yielding
$\textbf{venc} = [153.73, 25.62,30.75]^\intercal$\,cm/s. 
Because PRoM uses a 95.2\% larger 
$m_{13}$ thereby potentially leading to more intra-voxel dephasing, we advantage ODV and SDV by assuming no intra-voxel dephasing: $[s_{11},s_{21},s_{31}] = [20,20,20]$. We calculate RMSE averaged over $10^5$ trials at each true $v$ on $[-150, 150]$\,cm/s sampled every $0.5$\,cm/s. 

To simulate the complex intra-voxel dephasing and assess per-voxel estimator performance in this case, we simulate vessels as in \cite{loecher2018velocity} with circularly symmetric parabolic velocity profiles. Parameters include: $0.1$\, mm$^3$ isotropic resolution and $N_c = 1$ coil. The five vessels share the same peak velocity $60$\,cm/s but have different diameters of $5.5, 3.9, 3.2, 2.7, 2.4$\,mm. The proton density is set to be $30\%$ in the background region and $50\%$ in the static tissue region compared to that in the vessel regions. The complex signal is generated using \eqref{eq: datamodel} with symmetric three-point encoding such that $\textbf{venc} = [60, 20, 30]^\intercal$\,cm/s. Regions of $5\times 5 \times 5$ voxels are merged into one to generate intra-voxel dephasing and $0.5$\, mm$^3$ isotropic resolution. Then we add i.i.d.\ white complex Gaussian noise to make the maximum $s_{\alpha \beta}$ for all voxels in the vessel regions reach $30$. No post-processing is adopted for PRoM for pure comparison of per-voxel estimation performance in various dephasing scenarios.

\subsection{Phantom}
\label{subsection: Phantom}
A phantom experiment allows for a controlled comparison of estimation performance for SDV, ODV, PRoM, and PRoM+. An agarose gel-filled cylindrical container is used to generate the MRI signal. An air-coupled propeller rotates the container. The rotational rate is counted with a photomicrosensor~\cite{vali2020development}. The phantom was scanned on a $1.5$T scanner (Siemens MAGNETOM Avanto). The in-plane bottom to top velocity increases linearly with the horizontal component of distance from the center of the container. In this experiment, we encoded the vertical component of the in-plane velocity, which ranges from $-240$ to $240$\,cm/s, using symmetric a three-point acquisition. The acquisition parameters include: $N_c = 16$ coils; $\textbf{venc} = \left[250, 100, \frac{500}{3} \right]^\intercal$\,cm/s, following the recommended choice of venc ratio adopted in \cite{Schnell2017, Carrillo2019}; field-of-view (FOV) $520 \times 260$ mm; flip angle $5^\circ$; TR $4.38$ ms; TE $2.66$ ms; and, matrix size $192 \times 125$. There are $15$ repeated acquisitions. We use the averaged k-space as a reference to calculate the RMSE and aliasing error for all voxels except the background region across $15$ scans. For per-frame post-processing of PRoM, voxels with less than $30\%$ maximum voxel magnitude were masked out, and only the two most likely $v_{\bm{k}}$ were considered. Locally weighted quadratic surface class $\mathcal{U}$ using a span of 25\% closest points was adopted with $\lambda = 1$.

\subsection{In Vivo}
\label{subsection: In Vivo}
An in vivo experiment is used to verify that PRoM can unwrap velocity on an interval $\Omega$ larger than twice the largest venc, $\text{venc}_{21}$, as claimed in \eqref{eq: range vs largest venc} and to illustrate improved performance of PRoM+ using spatial information. A healthy volunteer was scanned on a $3$T scanner (Siemens MAGNETOM Vida). For the recruitment and consent of human subject used in this study, the ethical approval was given by an Internal Review Board (2005H0124) at The Ohio State University. The venc scouting scan showed the maximum absolute value of velocity to be above $90$\,cm/s and less than $100$\,cm/s. A breath-held, $N_c = 30$ coils, symmetric three-point encoding PC-MRI dataset was collected, with the imaging plane intersecting both the ascending aorta and descending aorta; through-plane velocity was encoded. Other acquisition parameters include: FOV $360 \times 270$ mm; flip angle $15^\circ$; TR $5.56$ ms; TE $3.69$ ms; matrix size $192 \times 108$; and, cardiac phases, $13$. The three-point acquisition is designed
using Alg.~\ref{alg: PRoM Optimal Design for Three-Point Encoding} for: $[P, Q] = [10,10], [s_{1\beta}, s_{2\beta},s_{3\beta}] = [5, 10, 5], \bm{\epsilon} = [10^{-6}, 10^{-6}]^\intercal, \Omega_{\bm{\epsilon}} = 200$\,cm/s, $\gamma m_\tau = \frac{\pi}{40}$\,s/cm. The design results in $\xi=5/3$
with highest first moment $\frac{\pi}{\gamma m_{13}} = 42$\,cm/s.
Restricted by input precision of the scanner interface, Alg.~2 gives
$\textbf{venc} = c [15,6,10]^\intercal$ for scaling $c=\tfrac{7}{2}$\,cm/s, yielding
$\textbf{venc} = \left[52.5, 21, 35\right]^\intercal$~cm/s.
The resulting unambiguous range is $\Omega = 210$\,cm/s, which is double the range $[-52.5, 52.5)$\,cm/s of SDV processing. For per-frame post-processing of PRoM, after square-root sum-of-squared coil combination, voxels less than $30\%$ maximum image were masked out, and only the two most likely $v_{\bm{k}}$ values were considered. Due to the complex velocity map across the FOV, locally weighted quadratic surface class $\mathcal{U}$ was adopted using a span of $3\%$ closest points and $\lambda = 1$.

\section{Results}
\subsection{Simulation Results}

\label{section: Results}
Fig.~\ref{fig: RMSE vs CRLB}  numerically explores the approximations adopted in PRoM.
The square root of the CRLB is plotted versus $s_{21}$ for the model in \eqref{eq: datamodel} and provides a lower bound on the RMSE for any unbiased estimator. The bound is from a local analysis of the likelihood function, and thus optimistically does not consider unwrapping errors. Superimposed are the RMSE values for both the MLE from the complex measurements and PRoM. Both estimators diverge from the bound for low SNR due to unwrapping errors. Both the MLE and PRoM asymptotically coincide with $\text{CRLB}^{0.5}$. Moreover, throughout the entire range of noise powers considered, the RMSE difference between MLE and PRoM is negligible, hence justifying the characterization of PRoM as an approximate MLE and validating the assumptions adopted in the derivation of the PRoM estimator.
\begin{figure}[h!]
    \centering
    \includegraphics[width = \columnwidth]{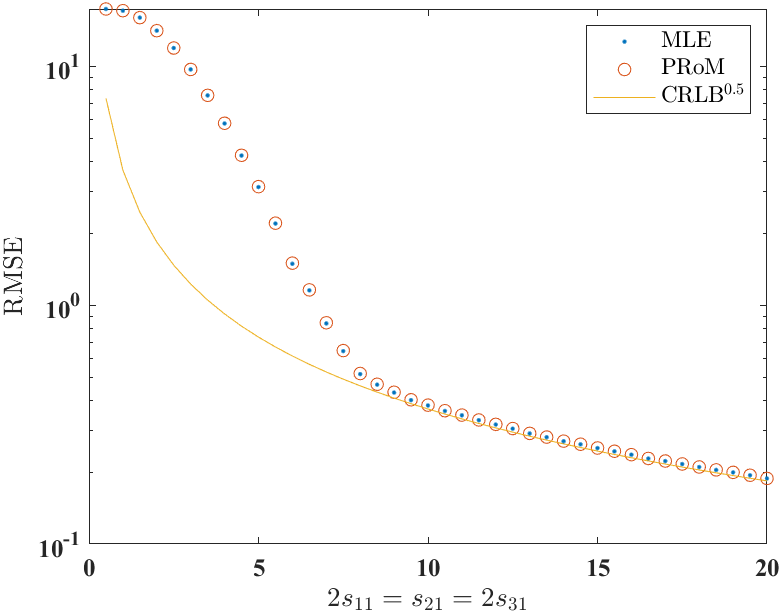}
    \caption{Comparison of PRoM, the MLE directly from the complex-valued voxels, and $\text{CRLB}^{0.5}$. RMSE is graphed versus $2s_{11} = s_{21} = 2s_{31}$; RMSE values for MLE and PRoM are averaged from $10^5$ trials.}
    \label{fig: RMSE vs CRLB}
\end{figure}

Fig.~\ref{fig: rmse_vs_range} shows the RMSE results for the same acquisition with $10^5$ trials at each true velocity value with simulation grid spacing $0.1$\,cm/s. Both ODV and PRoM can unwrap a large range of velocities. The bottom panel zooms into a smaller range of RMSE values; from a per-voxel estimation perspective, PRoM improves RMSE by modeling the non-zero noise correlation between phase difference measurements.

\begin{figure}[h!]
    \centering
    \includegraphics[width = \columnwidth]{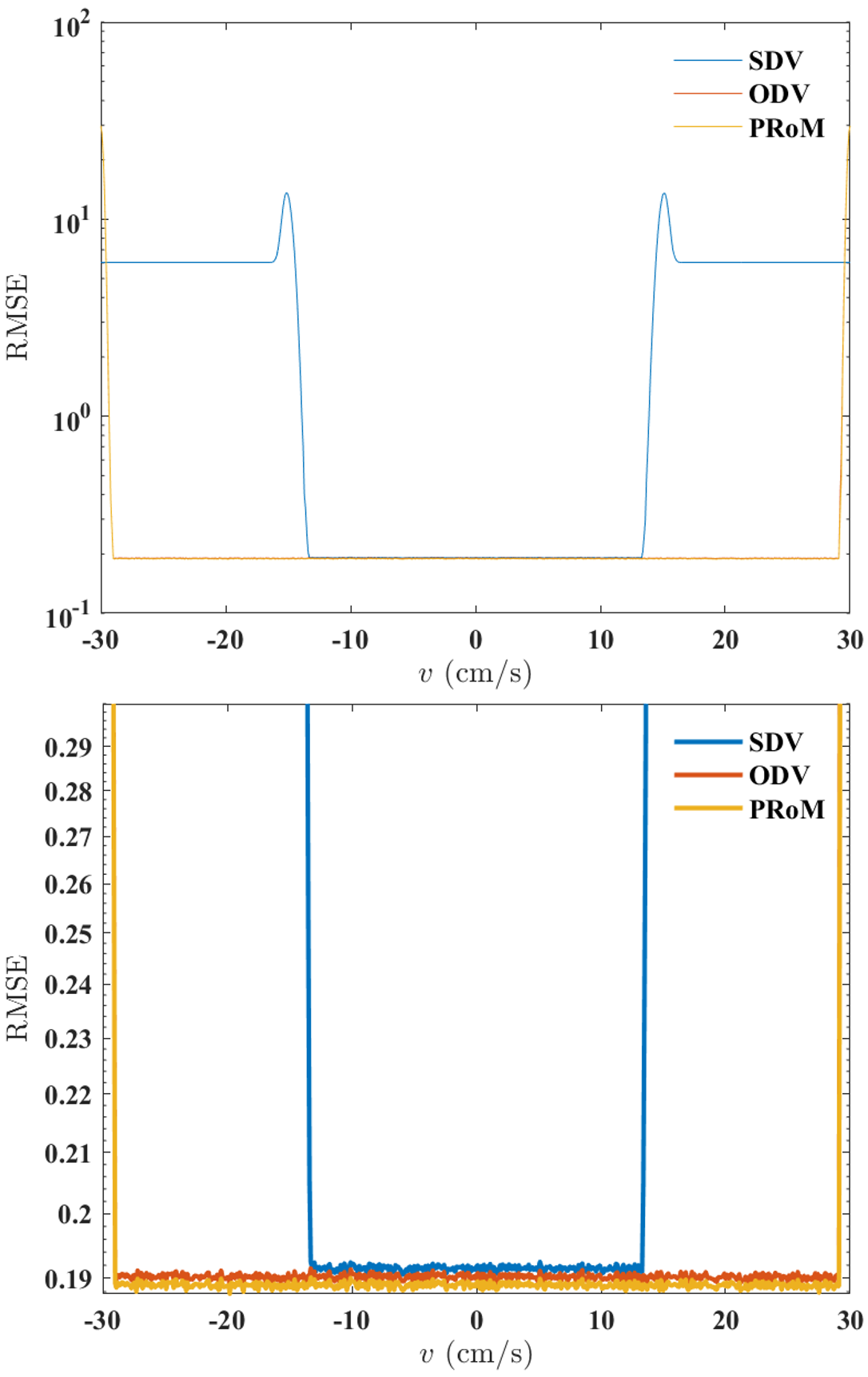}
    \caption{Top: RMSE versus true $v$ for $\textbf{venc} = [15, 6, 10]^\intercal$\,cm/s and $[s_{11}, s_{21}, s_{31}] = [10, 20, 10]$. RMSE values are averaged over $10^5$ trials. Bottom: zoomed-in version.}
    \label{fig: rmse_vs_range}
\end{figure}

 Fig.~\ref{fig: design comparision} shows the RMSE results for $10^5$ trials at each true velocity value with grid $0.5$\,cm/s. Here we advantage ODV and SDV by assuming no intra-voxel dephasing for their acquisition. For PRoM we assume intra-voxel dephasing leads to $2s_{11} = s_{21} = 2s_{31}$. The PRoM design uses $\xi=6/5$, which explicitly suppresses both Unwrapping Errors and Aliasing Errors to ensure reliable estimation across the full range, $[-150,150]$\,cm/s. Further, despite the handicap of simulated $50\%$ intra-voxel dephasing, PRoM still provides a 10.5\% decrease in RMSE compared to ODV and a 25.1\% decrease than SDV used with their suggested acquisition strategy.

\begin{figure}[h!]
    \centering
    \includegraphics[width = \columnwidth]{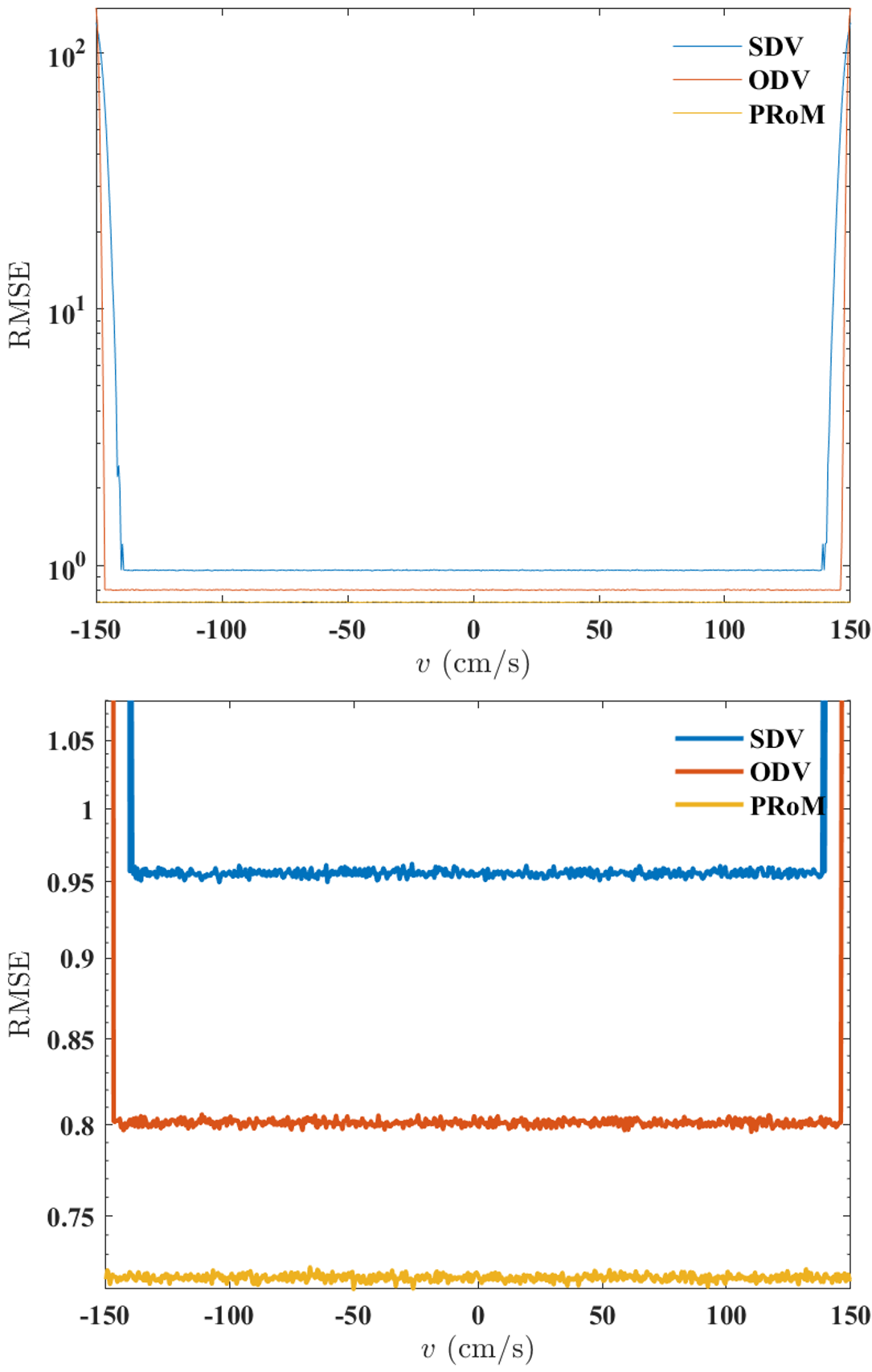}
    \caption{Comparison of acquisition design and estimation for SDV, ODV, and PRoM:  RMSE versus true velocity. Desired velocity range to be estimated is $[-150, 150]$\,cm/s. SDV uses three vencs generated are $\textbf{venc} = [150,60,100]^{\intercal}$\,cm/s. ODV uses three vencs are $\textbf{venc} = [150,50,75]^{\intercal}$\,cm/s. PRoM uses $\textbf{venc} = [153.73, 25.62,30.75]^\intercal$\,cm/s. Here we assume no intra-voxel dephasing for ODV and SDV $[s_{11}, s_{21}, s_{31}] = [20, 20, 20]$. For PRoM we assume intra-voxel dephasing and $[s_{11}, s_{21}, s_{31}] = [10, 20, 10]$. Top: RMSE versus true $v$.  RMSE values are averaged over $10^5$ trials. Bottom: zoomed-in version.}
    \label{fig: design comparision}
\end{figure}

Figure.~\ref{fig: Simulation of vessels} shows the per voxel results for simulation of intra-voxel dephasing. Panel (a) is the velocity profile simulated with refined resolution; (b) is the lower acquired resolution which leads to intra-voxel dephasing. Panels (c), (d), (e) illustrate intra-voxel dephasing for $m_{11},m_{12},m_{13}$. We observe more serious intra-voxel dephasing in $m_{11}, m_{13}$ and at voxels close to the boundaries of the simulated vessels. Panels (f), (g), (h) show the recovered velocity. For the flow region, we can observe aliasing error in the SDV estimated velocity, but not in ODV and PRoM. The RMSE values for SDV, ODV, and PRoM are $120.27, 10.24$, and  $10.12$\,cm/s, respectively. Moreover, the RMSE values given no Aliasing Error for SDV, ODV, and PRoM are $10.23, 10.24$, and  $10.12$\,cm/s. 
\begin{figure}[h!]
    \centering
    \includegraphics[width = \columnwidth]{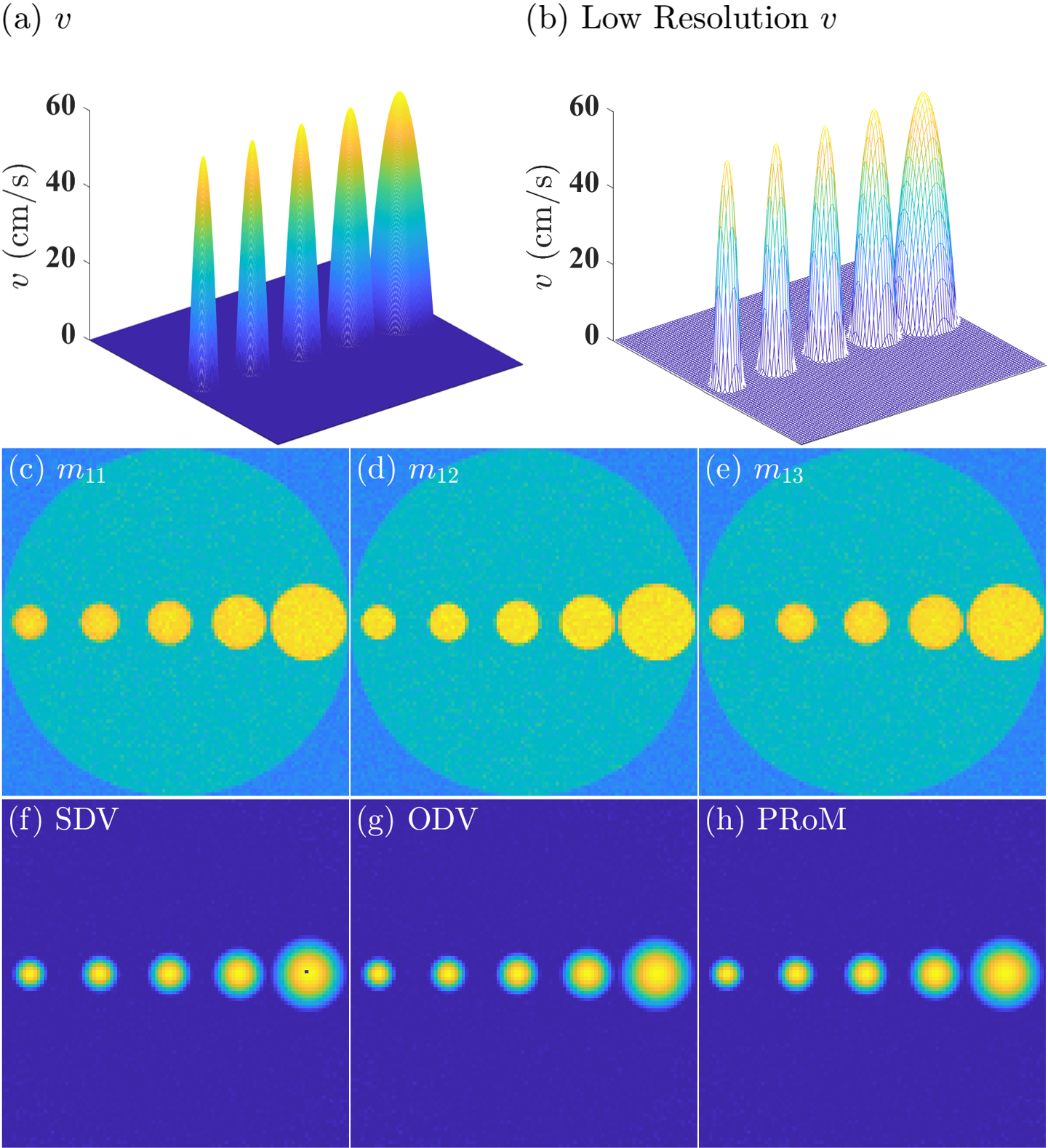}
    \caption{Comparison of SDV, ODV, and PRoM for simulated vessels. $\textbf{venc} = [60,20,30]^\intercal$\,cm/s. (a) is the true velocity profile. (b) is the lower acquisition resolution velocity profile. (c), (d), (e) are the magnitude images for $m_{11}, m_{12}, m_{13}$. (f), (g), (h) are the recovered velocities for SDV, ODV and PRoM.}
    \label{fig: Simulation of vessels}
\end{figure}

\subsection{Phantom Results}
\textcolor{black}{Fig.~\ref{fig: Phantom} uses measured data from a spinning phantom to evaluate RMSE in velocity estimation. In addition, the phantom data validate that both ODV and PRoM paired with simple post-processing can reliably unwrap a larger range of velocities than SDV, given the same encodings. The number of aliased voxels and RMSE values are reported in Table~\ref{tab: phantom}. The PRoM+ iteration is observed for all frames to converge in only two iterations. In this instance, PRoM+ eliminates aliasing errors and reduces RMSE by $25.8\%$ versus ODV, and $48.5\%$ compared to SDV.}
\begin{table}[h!]
\begin{center}
\caption{Comparison, for 1.5\,T phantom data, of estimator root mean squared error (RMSE) for SDV, ODV, PRoM, and PRoM+.}
\begin{tabular}{| l | r|r|r|r| }
    \hline
    Metric & SDV & ODV & PRoM & PRoM+ \\ \hline
    Number of aliased voxels & 27 &  5 &  5 & 0 \\ 
    RMSE of all voxels (cm/s) & 6.08 & 4.22& 3.85&  3.13\\ 
    RMSE excl. aliased voxels (cm/s) & 3.22& 3.59& 3.13&  3.13\\\hline
\end{tabular}
\label{tab: phantom}
\end{center}
\end{table}
\begin{figure}[h!]
    \centering
    \includegraphics[width = \columnwidth]{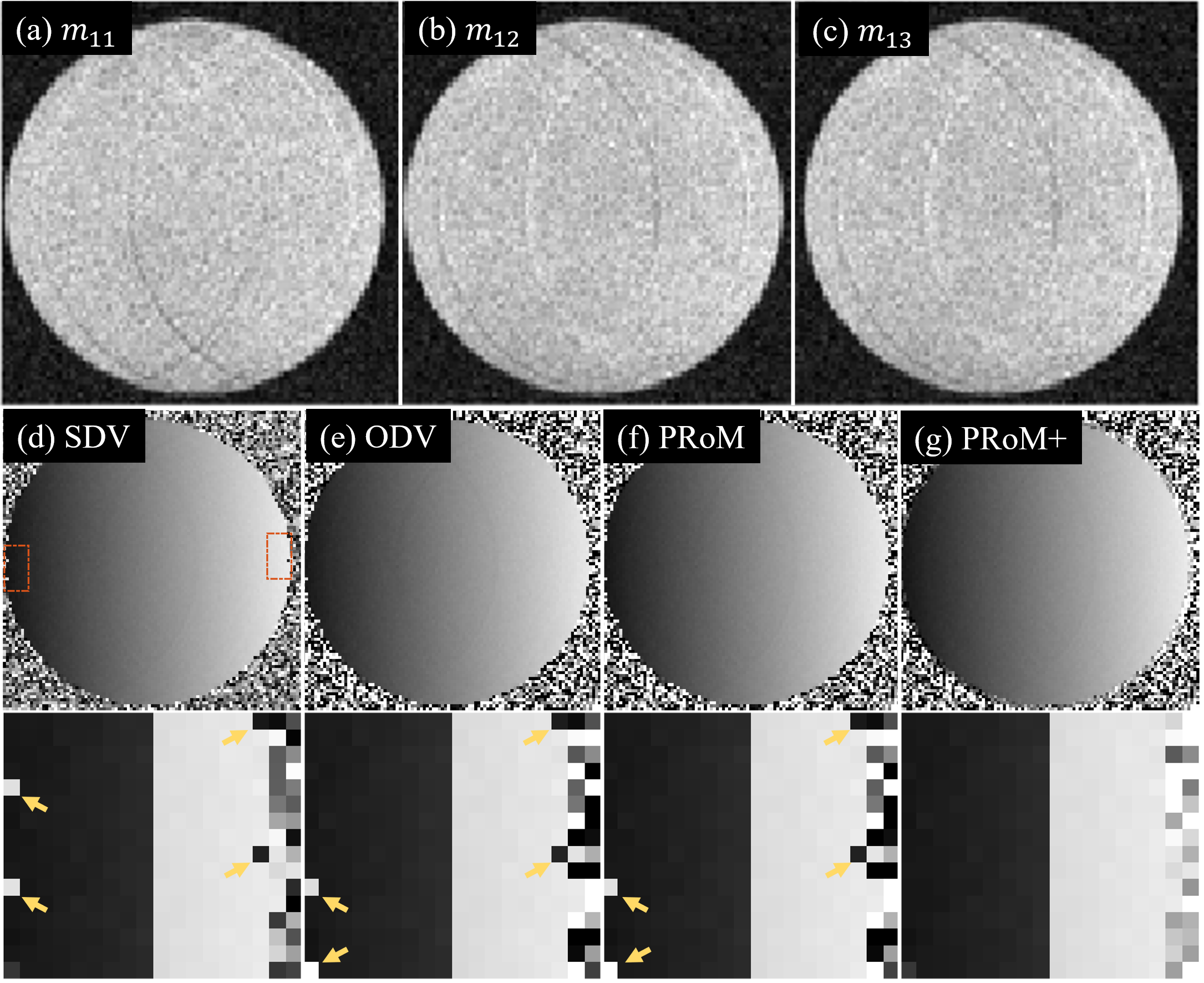}
    \caption{Comparison of estimators for SDV, ODV, PRoM, and PRoM+ using locally weighted quadratic surface class $\mathcal{U}$. $\textbf{venc} = \left[250, 100, \frac{500}{3}\right]^\intercal$\,\,cm/s. The in-plane speed is within $\pm 240$\,cm/s. (a), (b), (c) are the square root of sum of squared coil images for $m_{11},m_{12},m_{13}$. (d), (e), (f), (g) are the velocity estimates from SDV, ODV, PRoM, and PRoM+. Bottom row are the zoomed-in versions of (d), (e), (f), (g).}
    \label{fig: Phantom}
\end{figure}

\subsection{In Vivo Results}
\textcolor{black}{From Fig.~\ref{fig: In VIVO} panels (a), (b) and (c),  we can observe more serious intra-voxel dephasing for $m_{11}, m_{13}$ compared to $m_{12}$. Because the largest absolute value of true velocity is larger than the largest venc $52.5$\,cm/s, we observe significant aliasing in SDV recovery from (d). However, (e) and (f) illustrate that both ODV and PRoM can recover velocities larger than the largest $\text{venc}$. 
Here, the acquisition designed for PRoM departs from the $\xi=3/2$ heuristic to use $\xi=5/3$, resulting in an unambiguous range of velocities four times that for standard dual-venc for the same highest venc. (g) illustrates that PRoM+ can incorporate spatial correlations to improve unwrapping performance.}

\textcolor{black}{For the in vivo example using coil-combined image: ODV computation takes 9.106 seconds, and PRoM takes 2.246 seconds, with an additional 1.717 seconds for PRoM+. The iteration of PRoM+ is observed to converge in only two iterations for all frames.}

\begin{figure}[h!]
    \centering
    \includegraphics[width = \columnwidth]{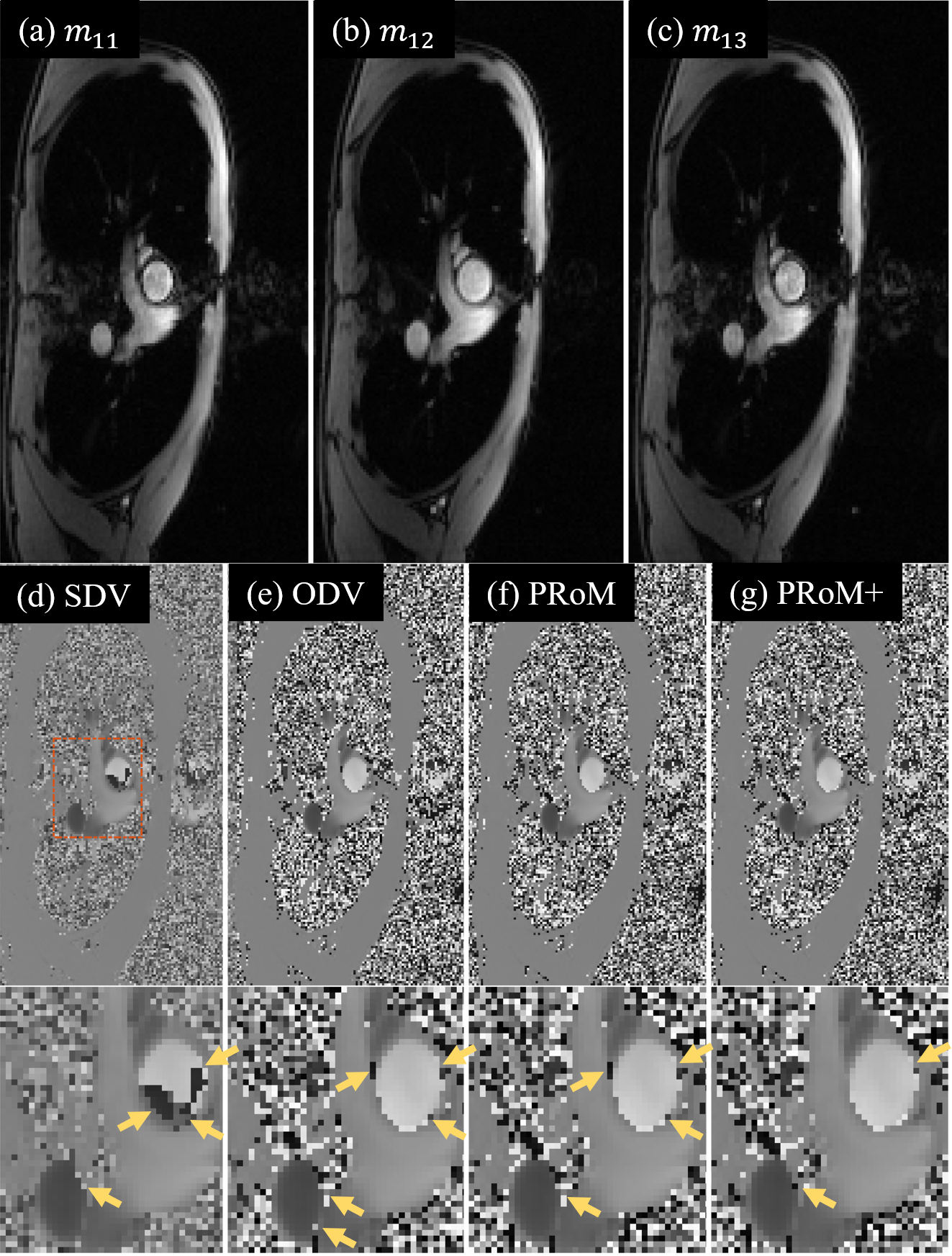}
    \caption{Comparison for SDV, ODV, PRoM, and PRoM+ using locally weighted quadratic surface class $\mathcal{U}$ velocity estimation for $\textbf{venc} = \left[52.5, 21, 35\right]^\intercal$\,cm/s. (a), (b), (c) are the square root of sum of squared coil images for $m_{11},m_{12},m_{13}$. (d), (e), (f), (g) are the velocity estimates from SDV, ODV, PRoM, and PRoM+. Bottom row are the zoomed-in versions of (d), (e), (f), (g).}
    \label{fig: In VIVO}
\end{figure}

\section{Discussion}
\label{section: Discussion}
\textcolor{black}{For the simulation and phantom studies, where the ground truth is available, PRoM offers a significant RMSE advantage over standard dual venc processing, i.e., 25.1\% for the simulation study and 48.5\% for the phantom study.  Although PRoM offers a fourfold computation advantage over ODV, its RMSE advantage over ODV, when both methods use the same venc values, is marginal (Fig.~\ref{fig: rmse_vs_range}). However, there are two features that distinguish PRoM from ODV, NCO, and other dual-venc methods. First, PRoM allows for an optimized venc design, which can translate to a more significant reduction in RMSE, as evident by 10.5\% reduction in RMSE over ODV (Fig.~\ref{fig: design comparision}). Second, PRoM can leverage the conditional probabilities of different wrapping integers to enable a new mechanism for phase unwrapping.} 

\textcolor{black}{The presentation here for PRoM is limited to a single component of velocity; extension to the estimation of all three velocity components, and hence congruence equations in multiple variables, is considered in \cite{zhao2022maximizing}.} 


\textcolor{black}{Several three-point encoding \cite{Carrillo2019, ma2020efficient} have been proposed and validated for PC-MRI aiming to improve VNR or unambiguous velocity range. Although performance depends strongly on vencs, the selection of vencs has been based on heuristics. 
PRoM, for the first time, provides an avenue to optimize vencs for three-point encoding. Fig.~\ref{fig: In VIVO} demonstrates that the PRoM-inspired design procedure in Alg.~\ref{alg: PRoM Optimal Design for Three-Point Encoding} can provide an unambiguous velocity range more than four times as large as the highest venc. The acquisition in Fig.~\ref{fig: In VIVO} illustrates the derivation in \eqref{eq: range vs largest venc} and the associated design procedure in Alg.~\ref{alg: PRoM Optimal Design for Three-Point Encoding}. Indeed, the design procedure allows the range to grow to the greatest extent allowed by the presumed noise floor, which is given as an input to the design. The design then minimizes the predicted RMSE while meeting constraints on unwrapping errors, reliable range of velocities, and maximum first moment. In the regime of low SNR, the PRoM design yields $\xi=3/2$, coinciding with a conventional heuristic \cite{Carrillo2019}. The Alg.~\ref{alg: PRoM Optimal Design for Three-Point Encoding} design provides unwrapping and velocity range guarantees, given a noise floor and bound on the highest first moment. For any higher SNR encountered, the guarantee still holds, and the RMSE reduces according to \eqref{eq: rmse given k}.}

\textcolor{black}{For volumetric imaging applications with vast numbers of voxels, the processing speed of PRoM may provide a desirable benefit. The careful pruning of the set of candidate wrapping integers results in a fast estimator without expensive grid search or gradient-based iterative optimization. }

\textcolor{black}{PRoM+ provides a simple but effective post-processing strategy for PRoM, which only affects the choice of $\bm{k}^\star$ from a Bayesian perspective. It differs from conventional unwrapping algorithms that typically only allow $\pm 2\pi$ adjustment in the possibly wrapped phases \cite{itoh1982analysis}. And, the processing incorporates the relative conditional probabilities of different wrapping integers, which are available as a byproduct of the PRoM algorithm. This strategy can be easily generalized to multiple velocity components and high-dimensional imaging. Although PRoM+ includes both covariance computation and spatial post-processing, the total computation time of PRoM+ for the in vivo example nonetheless is less than one-half the computation time of ODV, thereby enabling advanced processing in the clinical workflow.}

\textcolor{black}{Due to the fast computation speed, ability to incorporate constraints (e.g., bounds on the largest first moment strength), and ability to better process complex-valued multi-coil MRI data, PRoM can be readily integrated into the clinical workflow for (i) patient-specific, optimal venc design (Alg. \ref{alg: PRoM Optimal Design for Three-Point Encoding}), (ii) joint processing of the multi-point acquisition (Alg. \ref{alg: PRoM}), and (iii) enhanced phase unwrapping using PRoM+. The RMSE advantage of PRoM can enable more accurate quantification of slow flow for investigating conditions such as dilated aorta \cite{schnell2016improved} or left atrial blood stasis \cite{nakaza2021dual}. Moreover, the augmented phase unwrapping capabilities provided by PRoM+ improve recovery of peak velocities for in vivo data, where the voxels may have severe intra-voxel dephasing.}

\section{Conclusion}
\label{section: Conclusion}
\textcolor{black}{In this work, we propose PRoM, an algorithm to solve a noisy set of linear congruence equations. We apply PRoM to single dimension velocity recovery in multi-coil phase-contrast MRI, presenting results for three-point acquisition. PRoM provides a fast approximate maximum likelihood estimator that fully leverages all pairwise phase differences while seamlessly accommodating coil combining and amplitude attenuation due to dephasing. PRoM can recover velocities across the full unambiguous range, which can be much larger than twice the highest venc. Through PRoM, we can directly compute the probabilities of unwrapping errors and formulate the velocity estimate's probability distribution. This innovation, in turn, allows for the optimized design of the phase-encoded acquisition, guaranteeing minimum estimation error subject to user-defined constraints on desired velocity range, unwrapping errors, aliasing errors, and maximum first moment of the encoding gradient. Moreover, the wrapping error probabilities enable a simple but effective post-processing strategy for incorporating spatial correlations to further mitigate unwrapping errors. The processing does not require prior knowledge of sensitivity maps, dephasing, or per-voxel SNR; instead, an auto-tuning is provided by constructing a phase difference covariance matrix from the images across all encodings and coils. Simulation, phantom, and in vivo results validate the benefits of fast computation, reduced estimation error, increased unambiguous velocity range, and optimized acquisition.}


\section{Acknowledgment}
\label{section: Acknowledgement}
\textcolor{black}{The authors thank Dr.\ Ning Jin and Siemens Healthineers for providing the phantom data and Dr.\ Yingmin Liu and Dr.\ Chong Chen for their assistance with pulse sequence modification and data acquisition. The authors also thank Sizhuo Liu for her valuable discussion about post-processing.}

\newpage
\bibliographystyle{IEEEtran}
\textcolor{black}{\bibliography{PRoM}}
\end{document}